\documentclass[sigconf]{acmart}
\usepackage{makecell}
\usepackage{xspace}
\usepackage[skip=0pt]{subcaption}
\usepackage{listings}
\usepackage{multirow}
\usepackage{enumitem}
\usepackage{array}
\usepackage{bbm}
\usepackage{mathtools}
\usepackage{amsmath}
\usepackage[ruled,vlined,linesnumbered]{algorithm2e}

\setcopyright{cc}
\setcctype{by}
\acmJournal{PACMMOD}
\acmYear{2025} \acmVolume{3} \acmNumber{3 (SIGMOD)} \acmArticle{198} \acmMonth{6} \acmPrice{}\acmDOI{10.1145/3725335}

\newcommand*\LSTfont{\Small\ttfamily\SetTracking{encoding=*}{-60}\lsstyle}
\newenvironment{denseitemize}{
\begin{enumerate}[topsep=2pt, partopsep=0pt, leftmargin=1.5em]
  \setlength{\itemsep}{2pt}
  \setlength{\parskip}{0pt}
  \setlength{\parsep}{0pt}
}{\end{enumerate}}
\copyrightyear{2025}

\acmConference[SIGMOD '25]{ACM Management of Data}{June 03--05,
  2025}{Berlin, Germany}
\acmISBN{2836-6573/2025/6}

\begin{document}

\newcommand{\ie}{i.e.}
\newcommand{\eg}{e.g.}
\newcommand{\QF}{\textit{Query Failure}\xspace}
\newcommand{\asum}{\texttt{SUM}\xspace}
\newcommand{\aavg}{\texttt{AVG}\xspace}
\newcommand{\acount}{\texttt{COUNT}\xspace}
\newcommand{\minihead}[1]{\vspace{0.5em}\noindent\textbf{#1.}}
\newcommand{\todo}[1]{\textcolor{red}{(#1)}} 
\newcommand{\revise}[1]{\textcolor{blue}{#1}} 
\newcommand{\fix}[1]{\textcolor{blue}{#1}}
\newcommand{\oracle}{\name-O\xspace}
\newcommand{\uniform}{\name-R\xspace}
% properties
\newcommand{\aPrioriError}{P1\xspace}
\newcommand{\noMaintenance}{P2\xspace}
\newcommand{\notModifyDbms}{P3\xspace}
% Two stAge Query Approximation
\newcommand{\samplingscheme}{\texttt{TAQA}\xspace}
% Block Sampling for Approximate query Processing
\newcommand{\bstats}{\texttt{BSAP}\xspace}
\newcommand{\mycross}{\textcolor{red}{\textbf{$\times$}}}
\definecolor{forestgreen}{RGB}{10, 60, 10}
\newcommand{\mycheck}{\textcolor{forestgreen}{\textbf{\checkmark}}}
\definecolor{defaultcolor}{HTML}{1f77b4}
\newcommand{\syntaxhighlight}[1]{\textbf{\textcolor{defaultcolor}{#1}}}

\newcommand{\pilotrewriter}{\textsc{PilotRewriter}}
\newcommand{\finalrewriter}{\textsc{FinalRewriter}}
\newcommand{\pilotQ}{Q_{pilot}}
\newcommand{\finalQ}{Q_{final}}
\newcommand{\inQ}{Q_{in}}
\newcommand{\mysampleplan}{\Theta}
\newcommand{\myconstraint}{\Phi}

% SQL syntax
\newcommand{\query}{Q}
\newcommand{\sqltable}{T}
\newcommand{\sqlspace}{\;}
\newcommand{\sqldistinct}{\textbf{\texttt{DISTINCT}}\xspace}
\newcommand{\sqlselect}{\textbf{\texttt{SELECT}}\xspace\xspace}
\newcommand{\sqlfrom}{\textbf{\texttt{FROM}}\xspace\xspace}
\newcommand{\sqlwhere}{\textbf{\texttt{WHERE}}\xspace\xspace}
\newcommand{\whereclause}{\textbf{\texttt{WHERE}}\xspace}
\newcommand{\sqlgroupby}{\textbf{\texttt{GROUP}}\sqlspace\textbf{\texttt{BY}}\xspace}
\newcommand{\sqlhaving}{\textbf{\texttt{HAVING}}\xspace\xspace}
\newcommand{\sqlorderby}{\textbf{\texttt{ORDER}}\sqlspace\textbf{\texttt{BY}}\xspace\xspace}
\newcommand{\sqllimit}{\textbf{\texttt{LIMIT}}\xspace\xspace}
\newcommand{\sqlover}{\textbf{\texttt{OVER}}\xspace\xspace}
\newcommand{\sqlpartitionby}{\textbf{\texttt{PARTITION}}\sqlspace\textbf{\texttt{BY}}\xspace\xspace}
\newcommand{\sqljoin}{\textbf{\texttt{JOIN}}\xspace\xspace}
\newcommand{\sqlfulljoin}{\textbf{\texttt{FULL JOIN}}\xspace\xspace}
\newcommand{\sqlinnerjoin}{\textbf{\texttt{INNER}}\sqlspace\textbf{\texttt{JOIN}}\xspace\xspace}
\newcommand{\sqlleftjoin}{\textbf{\texttt{LEFT}}\sqlspace\textbf{\texttt{JOIN}}\xspace\xspace}
\newcommand{\sqlrightjoin}{\textbf{\texttt{RIGHT}}\sqlspace\textbf{\texttt{JOIN}}\xspace\xspace}
\newcommand{\sqlon}{\textbf{\texttt{ON}}\xspace\xspace}
\newcommand{\sqlas}{\textbf{\texttt{AS}}\xspace\xspace}
\newcommand{\sqlmax}{\textbf{\texttt{MAX}}\xspace\xspace}
\newcommand{\sqlmin}{\textbf{\texttt{MIN}}\xspace\xspace}
\newcommand{\sqlcount}{\textbf{\texttt{COUNT}}\xspace\xspace}
\newcommand{\sqlsum}{\textbf{\texttt{SUM}}\xspace\xspace}
\newcommand{\sqlavg}{\textbf{\texttt{AVG}}\xspace\xspace}
\newcommand{\sqlrank}{\textbf{\texttt{RANK}}\xspace\xspace}
\newcommand{\sqldenserank}{\textbf{\texttt{DENSE\_RANK}}\xspace\xspace}
\newcommand{\sqlunionall}{\textbf{\texttt{UNION ALL}}\xspace\xspace}
\newcommand{\mydots}{\cdots}
\newcommand{\assign}{:=}
\newcommand{\grammareq}{::=}
\newcommand{\aggregate}{a}
\newcommand{\target}{t}
\newcommand{\targetset}{\widetilde{t}}
\newcommand{\targetlist}{L}
\newcommand{\targetlistset}{\widetilde{\targetlist}}
\newcommand{\targetlists}{\widetilde{\targetlist}}
\newcommand{\sqlpredicate}{\psi}
\newcommand{\sqlcondition}{\sqlpredicate}
\newcommand{\sqlconditionset}{\widetilde{\sqlpredicate}}
\newcommand{\sqlgroupbyclause}{G}
\newcommand{\sqlorderbyclause}{O}
\newcommand{\sqlexpr}{E}
\newcommand{\sqlexprset}{\widetilde{\sqlexpr}}
\newcommand{\colalias}{alias\xspace}
\newcommand{\agg}{f}
\newcommand{\windowfunc}{\omega}
\newcommand{\col}{col}
\newcommand{\cols}{cols}
\newcommand{\gcols}{gcols}
\newcommand{\sqlvalue}{V}
\newcommand{\irule}[2]{
\mkern-2mu\displaystyle\frac{\begin{array}{c}#1\end{array}}{\vphantom{,}\begin{array}{c}#2\end{array}\mkern-2mu}}
\newenvironment{centermath}{\begin{center}$\displaystyle}{$\end{center}}

% proofs
\newcommand{\myjoin}{\mathcal{J}}
\newcommand{\name}{\textsc{PilotDB}\xspace}

\title{PilotDB: Database-Agnostic Online Approximate Query Processing with A Priori Error Guarantees (Technical Report)}

%%
%% The "author" command and its associated commands are used to define the authors and their affiliations.
\author{Yuxuan Zhu}
\affiliation{%
  \institution{University of Illinois}
  \city{Champaign-Urbana}
  \country{USA}
}
\email{yxx404@illinois.edu}

\author{Tengjun Jin}
\affiliation{%
  \institution{University of Illinois}
  \city{Champaign-Urbana}
  \country{USA}
}
\email{tengjun2@illinois.edu}

\author{Stefanos Baziotis}
\affiliation{%
  \institution{University of Illinois}
  \city{Champaign-Urbana}
  \country{USA}
}
\email{sb54@illinois.edu}

\author{Chengsong Zhang}
\affiliation{%
  \institution{University of Illinois}
  \city{Champaign-Urbana}
  \country{USA}
}
\email{cz81@illinois.edu}

\author{Charith Mendis}
\affiliation{%
  \institution{University of Illinois}
  \city{Champaign-Urbana}
  \country{USA}
}
\email{charithm@illinois.edu}

\author{Daniel Kang}
\affiliation{%
  \institution{University of Illinois}
  \city{Champaign-Urbana}
  \country{USA}
}
\email{ddkang@illinois.edu}

\begin{abstract}
After decades of research in approximate query processing (AQP), its adoption in 
the industry remains limited. Existing methods struggle to simultaneously provide 
user-specified error guarantees, eliminate maintenance overheads, and avoid 
modifications to database management systems. To address these challenges, we 
introduce two novel techniques, \samplingscheme and \bstats. \samplingscheme is 
a two-stage online AQP algorithm that achieves all three properties for 
arbitrary queries. However, it can be slower than exact queries if we use 
standard row-level sampling. \bstats resolves this by enabling block-level 
sampling with statistical guarantees in \samplingscheme. We implement 
\samplingscheme and \bstats in a prototype middleware system, PilotDB, that is 
compatible with all DBMSs supporting efficient block-level sampling. We evaluate 
\name on PostgreSQL, SQL Server, and DuckDB over real-world benchmarks, 
demonstrating up to 126$\times$ speedups when running with a 5\% guaranteed error.
\end{abstract}

\begin{CCSXML}
<ccs2012>
   <concept>
       <concept_id>10002951.10003227.10003241.10010843</concept_id>
       <concept_desc>Information systems~Online analytical processing</concept_desc>
       <concept_significance>500</concept_significance>
       </concept>
   <concept>
       <concept_id>10002950.10003648.10003688</concept_id>
       <concept_desc>Mathematics of computing~Statistical paradigms</concept_desc>
       <concept_significance>300</concept_significance>
       </concept>
   <concept>
       <concept_id>10002951.10002952.10003400</concept_id>
       <concept_desc>Information systems~Middleware for databases</concept_desc>
       <concept_significance>300</concept_significance>
       </concept>
   <concept>
       <concept_id>10002951.10003227.10003241.10003244</concept_id>
       <concept_desc>Information systems~Data analytics</concept_desc>
       <concept_significance>300</concept_significance>
       </concept>
 </ccs2012>
\end{CCSXML}

\ccsdesc[500]{Information systems~Online analytical processing}
\ccsdesc[300]{Mathematics of computing~Statistical paradigms}
\ccsdesc[300]{Information systems~Middleware for databases}
\ccsdesc[300]{Information systems~Data analytics}

  \keywords{approximate query processing, sampling}
  
  \received{9 April 2025}

\maketitle

\section{Introduction}
% introduce AQP and the most outstanding problem
Approximate query processing (AQP) is widely studied to accelerate queries in 
big data analytics \cite{aqua,start,blinkdb,abs,error-bounded-stratified,quickr,
sample+seek,idea,verdictdb,baq,taster,digithist,bias-aware-sketch,
count-filter-sketch,distributed-wavelet,icicles,congressional-sample,outlier,
sample-selection}. Although AQP has been extensively explored in academia, its adoption 
is still limited in practice \cite{sqlserver-aqp,oracle-aqp,bigquery-aqp}. Prior 
research demonstrates three properties that are crucial for real-world AQP 
applications: (\textbf{\aPrioriError}) guaranteeing user-specified errors before 
the query is executed (\ie, a priori error guarantees) \cite{quickr,blinkdb,baq,
taster,sample+seek,chaudhuri2017approximate,error-bounded-stratified}, 
(\textbf{\noMaintenance}) no maintenance overheads \cite{quickr,dbest}, and 
(\textbf{\notModifyDbms}) not modifying the underlying database management system 
(DBMS) \cite{verdictdb,dbest,baq,aqp-commercial-challenges}.

% status of existing literature
However, none of the existing systems or algorithms achieves all three properties simultaneously 
(Table \ref{tab:systems}). We can further categorize these techniques into 
two types: \textit{offline} methods that pre-compute samples and \textit{online} 
methods that generate samples at query time. 

% what is offline methods and why they fail
Existing offline AQP methods requires maintenance overheads \cite{blinkdb,
verdictdb,baq,aqua,icicles,congressional-sample,start,outlier,sample-selection,
sample+seek}, sacrificing \noMaintenance and leading to limitations in 
deployments. Offline methods operate in two stages. In the offline stage, they 
pre-compute data samples based on expected workloads. In the online stage, offline 
samples that satisfy the error specification are selected
to answer the query. Consequently, when data or workloads are updated, re-computations 
and/or manual inspections are required to maintain a priori error guarantees 
\cite{blinkdb,baq,sample+seek}. The cumulative costs of this maintenance can be 
a significant overhead that discourages deployments and commercial adoption 
\cite{chaudhuri2017approximate,mozafari2015cliffguard}.

% what is online methods and why they fail
Although online methods eliminate maintenance overheads (\noMaintenance) 
\cite{quickr,online-aggregation,ripple-join,dbo,gola,deepola,progressivedb},
existing online AQP algorithms require modifying DBMSs to achieve a priori error
guarantees \cite{quickr}, sacrificing \notModifyDbms. These algorithms depend on 
sophisticated samplers and complex optimization logic for query acceleration and 
error guarantees \cite{quickr}. However, these techniques are tightly integrated 
with DBMSs and lack widespread support. Consequently, adopting them requires 
modifying existing DBMSs, which can be infeasible for 
commercial applications \cite{aqp-commercial-challenges,verdictdb,keebo}.

% our vision
In this paper, we propose two novel techniques to simultaneously achieve \aPrioriError,
\noMaintenance, and \notModifyDbms, while accelerating queries compared to executing 
exact queries. First, we introduce a two-stage online AQP algorithm, 
\samplingscheme, that achieves a priori error guarantees through query rewriting 
and online sampling. Second, to accelerate query processing with \samplingscheme, 
we develop \bstats, a set of statistical techniques that formalize block sampling 
to provide statistical guarantees in approximate queries. Finally, we build a 
middleware AQP system, \name, which implements \samplingscheme and \bstats.

\begin{table}
    \centering
    \footnotesize
    \caption{Characteristics of state-of-the-art AQP systems and algorithms. 
    Online AQP inherently eliminates sample maintenance overhead. \name is the 
    first one that achieves a priori error guarantees  (\aPrioriError), 
    eliminates maintenance overheads (\noMaintenance), and avoids DBMS 
    modifications (\notModifyDbms), simultaneously. }\label{tab:systems}
    \begin{tabular}[pos]{c|c|c|c}
\toprule
AQP System & \makecell{A Priori Error\\ Guarantees (\aPrioriError)}
& \makecell{w/o Maintenance\\Overhead (\noMaintenance)} 
& \makecell{w/o Modifying\\DBMSs (\notModifyDbms)}  \\
\midrule
BlinkDB \cite{blinkdb}          & \mycheck & \mycross & \mycross  \\
Taster \cite{taster}            & \mycheck & \mycross & \mycross  \\
Sample+Seek \cite{sample+seek}  & \mycheck & \mycross & \mycross  \\
Quickr \cite{quickr}            & \mycheck & \mycheck & \mycross  \\
BAQ \cite{baq}                  & \mycheck & \mycross & \mycheck  \\
VerdictDB \cite{verdictdb}      & \mycross & \mycross & \mycheck  \\
DBest \cite{dbest}              & \mycross & \mycheck & \mycheck  \\
\textbf{PilotDB}                & \mycheck & \mycheck & \mycheck  \\
\bottomrule
    \end{tabular}
\end{table}

% introduce our algorithm
\minihead{\samplingscheme}
Our online AQP algorithm, \samplingscheme, achieves all three properties through 
two stages of query rewriting and online sampling. In the first stage, we 
rewrite the input query and execute it to determine the optimal sampling plan 
that (1) satisfies the user's error specification and (2) minimizes the execution 
cost. In the second stage, we rewrite the input query with the optimal sampling 
plan and execute it, delivering results directly to users. For both stages, the 
rewritten queries leverage existing samplers in the DBMS.

% existing sampling methods are not efficient for our algorithm
\vspace{0.5em}
However, naively applying samplers of existing work to \samplingscheme 
either fails to accelerate queries or requires modifying DBMSs. Specifically, 
row-level samplers are inefficient in databases that read data at the block 
level, resulting in query latencies as high as exact queries (\S 
\ref{subsec:block-motiv}) \cite{page-selection,sqlserver-store}. This is 
especially the case for analytical queries where data scanning is often the 
primary latency bottleneck \cite{chen2016memsql,theodoratos1997data}. To address 
the inefficiency of row-level samplers, researchers have developed more efficient 
sampling techniques, such as index-assisted sampling \cite{wang2023approximate,
zhao2022ab}. Unfortunately, these techniques require modifying DBMSs and are not 
widely supported, sacrificing \notModifyDbms.

\minihead{\bstats}
We develop \bstats to accelerate online sampling using block-level sampling (also
known as block sampling)—an approach that does not require modifying DBMSs (\notModifyDbms),
as it is already widely implemented \cite{impala-sample,presto-sample,psql-sample,
sqlserver-sample,duckdb-sample,hive-sample,bigquery-sample}. Block sampling, 
which samples data at the block level, achieves higher efficiency compared to 
row-level sampling by skipping non-sampled blocks (Figure \ref{fig:sampling-demo}).\footnote[1]{Throughout the paper, ``block'' refers to the minimum unit of data 
accessing in the storage layer.} Quantitatively, sampling 0.01\% data from a table 
with 6B rows using block sampling can be up to 500$\times$ faster 
than uniform row-level sampling (Figure \ref{fig:motivate-block}). Furthermore, 
our analysis reveals that for the same error specification, the sample size 
required for uniform block sampling can be comparable or even smaller than that of 
uniform row-level sampling (\S \ref{subsec:block-motiv}).\footnote[2]{In an extreme 
case where the variance and expectation of a block is similar to the entire dataset, 
sampling one block can be sufficient for a small error rate.}

\begin{figure}[t]
    \centering
    \noindent\includegraphics[width=\linewidth]{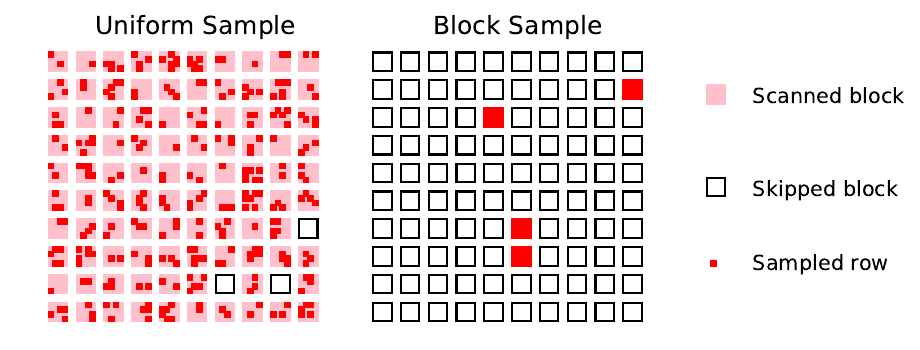}
    \caption{Sampling 3\% data from a table with a block size of 100 
    rows: in expectation, the row-level method requires scanning approximately 
    95\% blocks, while the block-level method scans approximately 3\% blocks.}
    \label{fig:sampling-demo}
\end{figure}

Although block sampling has been included in the ISO standard SQL \cite{sql-2003} 
and is widely supported, existing error analysis techniques are insufficient to
handle block sampling in nested or Join queries. Naively applying existing 
techniques can lead to errors up to 52$\times$ higher than the user-requested 
error (\S \ref{subsec:eval-correct}), sacrificing \aPrioriError. We introduce new 
statistical techniques in \bstats to formalize block sampling in approximate 
queries with statistical error guarantees. For deep nested queries, we establish 
sampling equivalence rules to reason about the interaction between block sampling 
and relational operations. For Join queries, we analyze the asymptotic distribution 
and the variance of the Join result over block samples, extending the standard 
central limit theorem (CLT) that fails when block sampling is executed on multiple 
tables \cite{huang2019joins,zhao2018random,chaudhuri1999random}.

With \bstats, we can further accelerate prior online AQP systems. In particular, 
we can use block sampling to replace the heavily-used uniform row-level sampling 
\cite{quickr}, while preserving their error guarantees. We empirically show that 
\bstats can accelerate \textsc{Quickr} by up to 60$\times$ (\S 
\ref{subsec:eval-augment}) and \samplingscheme by up to 219$\times$ (\S 
\ref{subsec:ablation}), compared to uniform row-level sampling.

\vspace{0.5em}

We build a prototype middleware AQP system, \name, that implements 
\samplingscheme and \bstats. We evaluate \name on various DBMSs, showing that it 
can achieve substantial query speedups on diverse benchmarks, including TPC-H 
\cite{tpch}, Star Schema Benchmark \cite{ssb-paper}, ClickBench \cite{clickbench}, 
Instacart \cite{instacart,verdictdb}, and DSB \cite{ding2021dsb}. When connected 
to transactional databases, PostgreSQL \cite{postgresql} and SQL Server 
\cite{sqlserver}, \name achieves up to 126$\times$ speedup. When connected 
to an analytical database, DuckDB \cite{duckdb}, \name achieves 
up to 13$\times$ speedup. \name consistently achieves a priori error 
guarantees across various settings. 

\vspace{0.5em}

We summarize our contributions as follows:
\begin{denseitemize}
    \item We propose \samplingscheme, an online AQP algorithm that 
    simultaneously achieves \aPrioriError, \noMaintenance, and \notModifyDbms
    (\S \ref{sec:tsqa}).
    \item We develop \bstats, a set of statistical techniques that enable block 
    sampling to answer approximate nested and Join queries with statistical 
    guarantees (\S \ref{sec:block}).
    \item We build and evaluate \name, which implements \samplingscheme and 
    \bstats, achieving a priori error guarantees and up to two orders of 
    magnitude speedup on various DBMSs (\S \ref{sec:eval}).
\end{denseitemize}

\section{Overview}
In this section, we present an overview of \name. We first discuss the background 
and challenges of building \name (\S \ref{subsec:bg-clg}). Next, we introduce the 
workflow of \name (\S \ref{subsec:overview-alg}). Finally, we describe the types 
of queries that \name supports (\S \ref{subsec:support}) and the semantics of 
errors that \name guarantees (\S \ref{subsec:error-semantic}).

\subsection{Background and Challenges}\label{subsec:bg-clg}
In Table \ref{tab:systems}, we summarize the characteristics of state-of-the-art 
AQP systems in terms of a priori error guarantees (\aPrioriError), maintenance 
overheads (\noMaintenance), and DBMS modifications (\notModifyDbms). We then 
present the background and challenges of simultaneously achieving \aPrioriError, 
\noMaintenance, and \notModifyDbms from the perspective of algorithmic and 
statistical techniques.

\minihead{Algorithmic Challenges}
Given a query and an error specification (\S \ref{subsec:error-semantic}), an AQP 
algorithm must plan sampling properly to achieve a priori error guarantees 
(\aPrioriError). A sampling plan specifies the sampling method, table(s) to sample, 
and the sample size, which determines whether the query can be accelerated and 
whether the error specification can be satisfied. To determine the sampling plan, 
prior work either pre-computes samples based on the knowledge of the query workload 
\cite{blinkdb,baq,taster,sample+seek} or inserts samplers to the query plan at query time 
based on runtime statistics \cite{quickr}. However, these methods 
break \noMaintenance or \notModifyDbms. The pre-computation method requires 
maintenance efforts to refresh samples when data changes \cite{blinkdb,baq,taster,sample+seek}, 
while the method of inserting samplers at query time requires modifying the 
execution and optimization logic of DBMSs \cite{quickr}.

We aim to resolve the tension among \aPrioriError, \noMaintenance, and \notModifyDbms. 
As we explained, the key challenge is to determine the sampling plan without 
pre-computation or controlling the query execution. To address it, we propose a 
novel online AQP algorithm that processes a query in two stages to plan and 
execute sampling (\S \ref{sec:tsqa}).

\minihead{Statistical Challenges}
Confidence intervals derived from statistical theories, such as CLT, are widely 
used to analyze errors of AQP \cite{blinkdb,verdictdb,quickr,baq,taster,
sample+seek,simple-random-sampling,sampling-algebra,group-by-sample,join-synopses,
aqua-sigmod,start,online-aggregation,ripple-join,dbo,turbo-dbo,block-sampling-count,
haas1996selectivity,haas2004bi}. However, deriving valid confidence intervals
for AQP with block sampling brings up two challenges that are not addressed in 
existing literature.

First, we need to analyze errors when there are intermediate relational operations
(\eg, Join and Group By) between block sampling and aggregations. Prior work 
studies confidence intervals of aggregations computed on the output of block 
sampling \cite{pansare2011online,haas1996selectivity,hou1991statistical}. However, 
with multiple relational operations between aggregations and block sampling, 
the confidence interval can potentially be affected by interactions between the 
block sampling and relational operations. Previous research on interactions 
between row-level sampling and relational operations cannot be applied to 
block-level sampling because they cannot handle the dependence of rows from the same 
block \cite{sampling-algebra,quickr}. In this work, we propose sampling 
equivalence rules that establish the commutativity between block sampling and 
relational operations (\S \ref{subsec:nested}), allowing us to analyze errors of 
deep approximate queries that use block sampling.

Second, we need to obtain valid confidence intervals when multiple tables in a 
Join query are sampled at the block level. Existing literature studies the 
asymptotic distribution of the Join result when tables are sampled with the same 
sample size \cite{haas1996selectivity}. However, it is insufficient in our case 
because we target a richer sampling space where sample sizes for tables can 
be arbitrarily different. To address that, we extend the theoretical result of 
Hass et al. \cite{haas1996selectivity} to a general form and derive an estimation 
of the upper bound of the sampling variance to achieve error guarantees (\S 
\ref{subsec:join}).

Those challenges are crucial to formalizing block sampling in AQP with error 
guarantees. We unify our theoretical results into \bstats, which can also be 
used to further accelerate other online AQP algorithms beyond \name 
(\S \ref{subsec:eval-augment}).

\subsection{Workflow}\label{subsec:overview-alg}
\begin{figure}
    \centering
    \includegraphics[width=\linewidth]{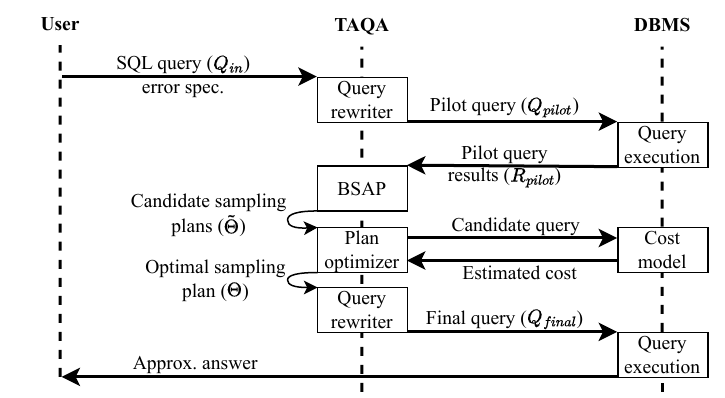}
    \caption{Workflow of \name.}
    \label{fig:archi}
\end{figure}
\name operates as a middleware between the user and the DBMS. Users may interact 
with \name in the same way as they interact with a DBMS except that \name takes 
additionally the error specification (\S \ref{subsec:error-semantic}) as input 
and produces an approximate answer.

On receiving the user's input, \name processes it with the \samplingscheme 
algorithm. \samplingscheme first rewrites the input query $\inQ$ into a pilot 
query $\pilotQ$ that computes necessary statistics for error analysis. Then,
\samplingscheme issues $\pilotQ$ to the DBMS and obtains the pilot result $R_{pilot}$.
Based on $R_{pilot}$ and the error specification, \samplingscheme incorporates 
\bstats to decide whether $\inQ$ can be approximated using block sampling. 
Specifically, \samplingscheme uses \bstats to analyze the error 
(\S \ref{sec:block}) and generates a set of candidate sampling plans, 
$\tilde{\Theta}$, that guarantees the error specification (\S \ref{subsec:plan}).
If \samplingscheme cannot identify any feasible sampling plans, \name will 
proceed to execute the original query $\inQ$.
Next, \samplingscheme interacts with the cost model of the DBMS to determine the 
optimal sampling plan $\Theta$ that minimizes the estimated cost 
(\S \ref{subsec:optimize}). Finally, \samplingscheme rewrites $\inQ$ in the final 
query $\finalQ$ based on $\mysampleplan$ and issues $\finalQ$ to the DBMS. We 
visualize this workflow in Figure \ref{fig:archi}.

\subsection{Supported Queries} \label{subsec:support}
\name is designed to answer all queries the underlying DBMS supports by directly 
executing the original query on the DBMS when necessary. There are two cases where 
\name may fail to accelerate a query: (1) \samplingscheme does not support the 
query or (2) block sampling cannot accelerate the query. In the first case, 
\name directly passes the query to the DBMS without intercepting it. In the second 
case, \name intercepts the query processing with \samplingscheme but executes 
the original query eventually.

\name tries to intercept and accelerate arbitrary aggregation queries using 
\samplingscheme with the following exceptions. First, \name does not support non-linear 
aggregates, (\eg, \texttt{COUNT DISTINCT}, \texttt{MAX}, and \texttt{MIN}), or 
aggregates in Group By clauses (\eg, \texttt{GROUP BY COUNT(*)}). These queries
are challenging for AQP and not supported in prior techniques \cite{blinkdb,
quickr,verdictdb,sample+seek}. Moreover, if any subqueries are correlated, \name 
tries to replace correlated subqueries with Joins using pre-defined rules. If 
\name fails to apply rules, it falls back to executing the exact query, since 
executing the pilot query is expensive if the query is correlated 
\cite{verdictdb}.

\name may fail to accelerate extremely selective queries or queries with a large 
group cardinality. These two cases are challenging to support for sampling-based 
AQP \cite{verdictdb,quickr}. However, prior online AQP may still use sampling on 
those queries, which results in errors larger than the user-specified error 
\cite{quickr}. By contrast, \name incorporates sampling plan 
optimization (\S \ref{subsec:optimize}) to determine that sampling is infeasible 
or not efficient for such queries. \name defaults to executing exact queries in 
this case.

\subsection{Error Specifications and Semantics} \label{subsec:error-semantic}
Finally, we describe how users can specify error requirements in \name and define 
the statistical semantics of the error specification. \name allows users to 
specify a maximum relative error of all aggregates and a probability or confidence, 
which are the same specifications prior work allows \cite{blinkdb,
error-bounded-stratified,quickr,taster}.  We present an example query with error 
specifications below:
\begin{lstlisting}[keywords={SELECT,COUNT,FROM,TABLESAMPLE,WHERE,SUM,AVG,AND,INNER,JOIN,GROUP,BY,INTERVAL,ERROR,WITHIN,PROBABILITY}]
-- example query
SELECT l_returnflag, l_linestatus, SUM(l_extendedprice) 
       as agg_1, AVG(l_extendedprice) as agg_2
FROM lineitem
WHERE l_shipdate <= date '1998-12-01' - interval '90 day'
GROUP BY l_returnflag, l_linestatus
-- error specification
ERROR WITHIN 5%
PROBABILITY 95%
\end{lstlisting}

Intuitively, the error specification in the example query means that the probability 
of relative errors of \texttt{agg\_1} and \texttt{agg\_2} being simultaneously 
less than 5\% is at least 95\%. Formally, consider a query with $k$ aggregations 
and $m$ groups, resulting in a set of $k\cdot m$ aggregates: 
$\{\mu_{i,j} | 1 \le i \le k, 1 \le j \le m\}$. We denote $\hat \mu_{i,j}$ as 
the estimate of the aggregate $\mu_{i,j}$. An error specification with a 
relative error $e$ and confidence $p$ means that \name will output a set of 
estimated aggregates such that the probability 
that all estimates simultaneously have a relative error no greater than $e$ 
(\ie, the probability of the intersection of events) is at least $p$. 
Namely,
\begin{equation}
\mathbb{P}\left[\bigcap_{\substack{1\le i \le k, 1\le j \le m}}\ 
    \left|\frac{\mu_{i,j} - \hat \mu_{i,j}}{\mu_{i,j}}\right| \le e 
    \right] \ge p \label{eq:error-semantics}
\end{equation}

Our error specification limits the joint probability of all estimates having 
unexpected errors across aggregations and groups. This is stronger and more 
intuitive for users to interpret than prior work that only reasons about the error 
of each estimate independently \cite{quickr,blinkdb}. We will tackle the joint 
probability in the next section.

\section{Two-Stage Query Approximation}\label{sec:tsqa}
In this section, we focus on addressing the algorithmic challenges mentioned in 
Section \ref{subsec:bg-clg}. We introduce our two-stage query approximation 
algorithm to answer the following three questions:
\begin{enumerate}[leftmargin=*]
\item How can we find a \textit{valid} sampling plan that satisfies the user's 
        error specification (\S \ref{subsec:plan})?
\item How can we find an \textit{efficient} sampling plan that minimizes the 
        execution cost of \samplingscheme (\S \ref{subsec:optimize})?
\item How can we achieve (1) and (2) via query rewriting (\S \ref{subsec:rewrite})?
\end{enumerate}

\subsection{Sample Planning via Pilot Query Processing} \label{subsec:plan}
We determine sampling plans that satisfy the user's error specification by 
executing a pilot query that inspects the statistical property of the input 
query. To understand what should be inspected through the pilot query, 
we first parametrize the sampling plan.

Given a query with $k$ tables, a sampling plan should specify the sampling 
method and corresponding sampling parameters for each table. To avoid
modifying the DBMSs, we use Bernoulli sampling 
where each unit (\eg, a row or a block) is independently selected with a fixed 
sampling rate or probability $\theta$ without replacement. In many DBMSs 
\cite{sql-2003}, row-level Bernoulli sampling is supported through the 
\texttt{TABLESAMPLE BERNOULLI} clause while block-level Bernoulli sampling is 
expressed via \texttt{TABLESAMPLE SYSTEM}. 

Although Bernoulli sampling produces variable sample sizes, we can still
provide error guarantees by parameterizing the $k$-table sampling plan into a 
list of $k$ sampling rates: $\Theta=[\theta_1, \ldots, \theta_k]$. This approach 
allows us to account for the variability in sample sizes when deriving 
guarantees. In the rest of this section, we present the statistical intuition 
and formulation underlying this approach.

\minihead{Statistical Intuition}
Consider the scenario where the query involves one aggregation computed on one 
group. We can calculate the confidence interval to analyze the relative error of 
the estimate. Suppose we have a population with mean $\mu$ that is estimated with 
a sample mean $\hat\mu$. We denote $Var[\hat\mu]$ as the variance of $\hat\mu$. 
We can establish the following CLT-based confidence interval for $\mu$:
\begin{equation}
    \mathbb{P}\left[
        \hat\mu - z_{(1+p)/2}\sqrt{Var[\hat\mu]} 
        \le \mu \le 
        \hat\mu + z_{(1+p)/2}\sqrt{Var[\hat\mu]}
    \right] \ge p \label{eq:ci}
\end{equation}
where $z_{(1+p)/2}$ is the $(1+p)/2$ percentile of the standard normal 
distribution. When $\mu$ is positive, Inequality \ref{eq:ci} can be equivalently 
converted to an inequality on the relative error of $\hat\mu$:
\begin{equation*}
    \mathbb{P}\left[
        \left|\frac{\hat \mu - \mu}{\mu}\right| \le 
        \frac{z_{(1+p)/2}\sqrt{Var[\hat\mu]}}{\mu}
    \right] \ge p
\end{equation*}
That is, to satisfy the error specification with a maximum relative error $e$ 
and a confidence $p$, it is sufficient to ensure that
\begin{equation}
    z_{(1+p)/2} \cdot \sqrt{Var[\hat\mu]} \cdot \mu^{-1} \le e \label{eq:rel-error-gt}
\end{equation}

With Inequality \ref{eq:rel-error-gt}, we observe that determining $\mu$ and 
$Var[\hat\mu]$ is the key to satisfying the error specification. However, $\mu$ and 
$Var[\hat\mu]$ are unknown unless we execute the input query. To address this, 
prior work maintains pre-computed samples \cite{blinkdb,baq,sample+seek} or 
modifies DBMSs to monitor statistics during the query execution \cite{quickr}.
In \samplingscheme, 
we estimate $\mu$ and $Var[\hat\mu]$ by executing a pilot query that is 
dynamically rewritten from the input query. 

To minimize the latency overhead, the pilot query samples the table that 
is most expensive to load. This is achieved in two steps. First, \name obtains 
an execution plan of the original query to inspect the table loading method used 
by the DBMS. A table is considered as a candidate to sample if the DBMS uses 
scanning.\footnote[3]{\textcolor{black}{Due to the overhead, sampling is often 
slower than index seeking, which is often used when the table is indexed and 
predicates are highly selective.}}
Second, \name queries the estimated table 
cardinality maintained by the DBMS and samples the largest table that will be 
scanned.

From the pilot query result, we can estimate the lower bound of $\mu$ 
and the upper bound of $Var[\hat\mu]$ where $\hat\mu$ will be computed using a 
sampling plan $\Theta$ in the final query. We first focus on sampling one table 
in the final query and then address sampling multiple tables in Section \ref{subsec:join}.
Assuming $\hat\mu$ is sub-Gaussian,\footnote[4]{\textcolor{black}{Sub-Gaussian 
assumption holds for any bounded distribution based on Hoeffding's inequality. 
Estimates of aggregate are bounded as tables have finite cardinality.}} these 
bounds are estimated using standard technique based on the
CLT \cite{intro-math-stats}, a widely used approach in AQP 
\cite{verdictdb,quickr,sample+seek,blinkdb}. The sub-Gaussian 
assumption implies that $\hat\mu$ has a fast decaying tail bounded above by a 
Gaussian distribution. Then, the analytical expression of the bounding 
distribution can be derived using CLT asymptotically. We present a detailed 
derivation in Appendix \ref{subsec:app-bounds}.

For instance, given sample size $n$ and sample variance $\hat\sigma$, we have 
$\frac{\mu - \hat\mu}{\hat\sigma/n} \rightarrow t_{n-1}$ as $n \rightarrow 
\infty$. With sufficiently large $n$, we have:
\begin{equation*}
    \mathbb{P}\left[\mu \ge \hat\mu - t_{n-1,1-\delta}\cdot \hat\sigma \cdot n^{-1}\right] \ge 1-\delta
\end{equation*}
where $\delta$ is a pre-specified failure probability and $t_{n-1,1-\delta}$
is $1-\delta$ percentile of Student's t distribution.
We can obtain the upper bound of $Var[\hat\mu]$ similarly since the ratio 
between the variance $\sigma^2$ and its estimate $\hat\sigma^2$ converges to 
chi-squared distribution: $(n-1)\hat\sigma^2/\sigma^2 \rightarrow \chi^2_{n-1}$.
Furthermore, as $n$ follows the binomial distribution $Bin(N,\theta)$, we 
can estimate the lower bound of $n$ given the upper bound of the population size 
$N$ that is obtained using the pilot query result.

However, this is not sufficient to guarantee the confidence $p$ since these 
bounds obtained from statistical distributions are \textit{probabilistic}. 
A probabilistic bound can fail with a \textit{controllable} probability 
\cite{intro-math-stats}. Therefore, to ensure the overall validity, we adjust 
the confidence $p$ based on the failure probability of all probabilistic bounds 
we used in the derivation, which leads to the confidence $p'$ in 
Procedure \ref{theorem:rel-error}.

\minihead{Formal Description} 
We formalize the intuition as follows.

\newtheorem{myProcedure}{Procedure}

\begin{myProcedure}\label{theorem:rel-error}
Consider an input query $\inQ$ that computes a linear aggregate $\mu$. Suppose
a user specifies a maximum relative error $e$ and a confidence $p$. In the first
stage, we rewrite $\inQ$ into a pilot query $\pilotQ$ with sampling rate 
$\theta_p$. Based on the result of $\pilotQ$, we can 
calculate (1) $L_\mu$: a probabilistic lower bound of $\mu$, and (2) 
$U_V[\Theta]$: a probabilistic upper bound of $Var[\hat\mu]$ given a sampling 
plan $\Theta$. Namely, with pre-specified failure probabilities $\delta_1$ 
and $\delta_2$, we can obtain the following inequalities:
\begin{gather}
    \mathbb{P}\left[\mu \ge L_\mu\right] \ge 1 - \delta_1 \label{eq:mean-lb}\\
    \mathbb{P}\left[Var[\hat\mu] \le U_V[\Theta]\right] \ge 1 - \delta_2 \label{eq:var-ub} 
\end{gather}
We find a sampling plan $\Theta$ such that the following inequality holds
\begin{equation}
    z_{(1+p')/2} \cdot \sqrt{U_V[\Theta]} \cdot L_\mu^{-1} \le e \label{eq:key-cstr}
\end{equation}
where $p'$ is the adjusted confidence based on the probabilities in Inequalities 
\ref{eq:mean-lb} and \ref{eq:var-ub}:
\begin{equation*}
    p' = p + \delta_1 + \delta_2 \notag
\end{equation*}
\end{myProcedure}

Procedure \ref{theorem:rel-error} involves three tunable parameters: 
$\theta_p$, $\delta_1$, and $\delta_2$. Intuitively, a smaller $\theta_p$ 
reduces overhead of executing $Q_{pilot}$, while a larger $\theta_p$ results in 
tighter estimations. Similarly, an optimal allocation of probabilities 
(configurations of $\delta_1$ and $\delta_2$) can lead to smaller sampling rates 
and thus higher query speedups. By default, we set $\theta_p=0.05\%$ and 
$\delta_1=\delta_2=1-p'=\frac{1-p}{3}$. In line with existing literature 
\cite{abae,inquest,quickr,hou1991error,hertzog2008considerations}, we recommend 
configuring $\theta_p$ to ensure that the pilot sample typically includes more 
than 30 units. For those requiring optimal performance, we suggest efficiently 
tuning $\delta_1$ and $\delta_1$ using cached pilot query results.

Following Procedure \ref{theorem:rel-error}, we can obtain an estimated 
aggregate $\hat\mu$ that satisfies the user's error specification. We formally 
state the guarantee in Theorem \ref{theorem:guarantee} and present the proof in
Appendix \ref{sec:app-proof}.

\begin{theorem}\label{theorem:guarantee}
Assuming that the aggregate to estimate is sub-Gaussian, if the input 
query $\inQ$ is
rewritten into a final query $\finalQ$ based on the sampling plan $\Theta$ obtained 
from the Procedure \ref{theorem:rel-error}, the estimated aggregate $\hat\mu$ 
computed in $\finalQ$ satisfies the inequality: 
$\mathbb{P}\left[\left|(\hat\mu - \mu)/\mu\right| \le e\right] \ge p$.
\end{theorem}

In \name, $L_\mu$ and $U_V[\Theta]$ cannot be naively obtained through standard 
techniques since \name uses block sampling, instead of row-level sampling. 
Block sampling introduces correlations among data from the same block, which 
breaks the assumption of data independence in standard techniques \cite{blinkdb,
quickr,verdictdb,intro-math-stats}. We develop a set of novel statistical 
techniques, \bstats, to address that (\S \ref{sec:block}).

\minihead{Multi-Aggregate Queries}
It is common to calculate more than one aggregate in a single query by computing 
arithmetic combinations of multiple aggregations, specifying multiple aggregations, 
or grouping a table by columns. To guarantee the overall error specification on 
all aggregates, we need to adjust the error requirement (\ie, the relative error 
$e$ and the confidence $p$) for each aggregate.

\begin{table}
    \centering
    \caption{Upper bounds of relative errors of composite estimators with multiplication, division, and addition.}
    \label{tab:error-prop}
    \begin{tabular}{cc}
    \toprule
    \makecell[c]{Composite\\estimator} & Upper bound of relative error \\
    \midrule
    $\hat \mu_1 \cdot \hat \mu_2$ & $e_{\mu_1} + e_{\mu_2} + e_{\mu_1} \cdot e_{\mu_1}$ \\
    $\hat \mu_1 / \hat \mu_2$ & $(e_{\mu_1}+e_{\mu_2})/(1+\min(e_{\mu_1}, e_{\mu_2}))$ \\
    $\hat \mu_1 + \hat \mu_2$ & $\max(e_{\mu_1}, e_{\mu_2})$\\
    \bottomrule
    \end{tabular}
\end{table}
First, we discuss how \samplingscheme deals with composite aggregates that 
compute (nonlinear) arithmetic combinations of simple aggregates, such as the 
product of two \asum aggregates. In \samplingscheme, we handle composite 
aggregates by propagating the relative error of simple aggregates (\eg, the sum
aggregates) into the composite aggregates (\eg, the product). In the case of 
estimating the product of two simple aggregates, the relative error of the 
product can be bounded above by the relative errors of the factors:
\begin{equation}
    \left|\frac{\hat\mu_1 \cdot \hat\mu_2 - \mu_1 \cdot \mu_2}{\mu_1 \cdot \mu_2}\right| 
    \le \left|\frac{\hat\mu_1 - \mu_1}{\mu_1}\right|\left|\frac{\hat\mu_2 - \mu_2}{\mu_2}\right| 
    + \left|\frac{\hat\mu_1 - \mu_1}{\mu_1}\right| + \left|\frac{\hat\mu_2 - \mu_2}{\mu_2}\right| 
    \notag
\end{equation}
This inequality shows that it is sufficient to limit the relative error of 
factors for the relative error of the product to satisfy the error specification. 
In \name, we allocate the relative error requirement evenly across simple 
aggregates. Therefore, each simple aggregate will need to satisfy a relative 
error of $e'=\sqrt{e+1} - 1$.

We refer to this way of using the relative error of simple aggregates to limit 
the relative error of a composite aggregate as \textit{error propagation}. 
We introduce propagation rules for multiplication, division, and addition in 
Table \ref{tab:error-prop}, which are inspired by uncertainty propagations 
\cite{error-analysis,error-prop-multi,uncertainty-prop}. The validity of 
these rules can be proved with straightforward algebraic transformation. We 
present the detailed proof in Appendix \ref{sec:app-proof}.

Second, in the case where a query computes multiple aggregates, \samplingscheme 
adjusts the confidence $p$ and applies the procedures in Procedure 
\ref{theorem:rel-error} to each of them. Based on our error semantics (\S 
\ref{subsec:error-semantic}), \samplingscheme should guarantee that the joint 
probability of the relative error of each estimate being less than $e$ is at 
least $p$. To analyze the joint probability, we apply Boole's inequality, which 
decomposes the probability of a union of events into the sum of probabilities of 
individual events:
\begin{align*}
\mathbb{P}\left[\bigcap_{\substack{1\le i \le k, 1\le j \le m}} \hat e_{i,j} \le e \right]
&= 1 - \mathbb{P}\left[\bigcup_{\substack{1\le i \le k, 1\le j \le m}} \hat e_{i,j} \ge e \right] \\
&\ge 1 - \sum_{i=1}^k \sum_{j=1}^m \mathbb{P}\left[\hat e_{i,j} \ge e\right]
\end{align*}
where $\hat e_{i,j} = |(\mu_{i,j} - \hat\mu_{i,j})/\mu_{i,j}|$ is the relative 
error of the aggregate estimate $\hat\mu_{i,j}$. This inequality shows that it 
is sufficient to limit the summation of the confidence of individual aggregates 
for the overall confidence to hold. With such decomposition, we can conveniently 
allocate the confidence to each aggregate. In \name, we allocate the confidence 
evenly. Namely, if we have $k\cdot m$ aggregates, each aggregate $\mu_{i,j}$ 
needs to satisfy its relative error requirement with confidence of 
$p_{i,j} = 1-\frac{1-p}{km}$.

\minihead{Handling Missing Groups}
Till now, we have been focusing on analyzing the error of estimations. However, 
for queries with Group By clauses, it is possible to miss groups in the pilot 
query due to block sampling. In this case, we may result in a sampling plan that 
does not guarantee errors of aggregates of missed groups. To address it, 
\samplingscheme controls the sampling rate of the pilot query to ensure that 
groups larger than a user-specified value $g$ are not missed with a high
probability. If all groups output by the query are smaller than $g$, 
\samplingscheme will end up generating a sampling plan with large sampling 
rates, making the approximate query more expensive than the original query. Such 
sampling plans will be rejected during the sampling plan optimization 
(\S \ref{subsec:optimize}). Consequently, \name will execute these queries 
exactly.

To ensure that all groups with size greater than $g$ are included in 
the pilot query results 
with a high probability, we propose the following lemma that computes the 
required sampling rate of the pilot query. We present the proof of the 
lemma in Appendix \ref{sec:app-proof}.
\begin{lemma}\label{lemma:group-error}
For a table $T$ with a block size $b$, block sampling with a sampling rate 
$\theta$ satisfying the condition below ensures that the probability of missing 
a group of size greater than $g$ is less than $p_f$.
\begin{equation}
    \theta \ge 
    1 - \left(1-\left(1-p_f\right)^{\lceil g/b \rceil/|T|}\right)^{1/\lceil g/b \rceil} 
    \label{eq:group-error}
\end{equation}
\end{lemma}

Intuitively, Lemma \ref{lemma:group-error} calculates the minimum sampling rate
to maintain a high group coverage probability. This result extends the group
coverage probability of row-level sampling in prior work (\ie, Proposition 4 of 
\cite{quickr}) to block sampling. Empirically, with $g=200$ and $p_f=0.05$, no 
groups are missed for the queries we evaluated (\S \ref{sec:perf}). Nevertheless, 
there is an opportunity to integrate block sampling with indexes, such as the 
outlier index \cite{outlier}, to better support small-group queries, left to 
future work.

\subsection{Sampling Plan Optimization} \label{subsec:optimize}
For queries with multiple input tables, Procedure \ref{theorem:rel-error} often 
results in multiple valid sampling plans. \samplingscheme uses optimization 
methods to find the most efficient plan. We formulate sampling 
plan optimization as a mathematical optimization problem and derive a solution 
using cost models.

\minihead{Problem Formulation}
According to Procedure \ref{theorem:rel-error}, the error specification is satisfied 
if the sampling plan satisfies each constraint $\phi_{i,j}$ of $i$-th aggregation 
and $j$-th group, as defined below:
\begin{equation}
    \phi_{i,j}(\Theta) :\equiv \quad  z_{(1+p_{i,j})/2} \cdot \sqrt{U_{V_{i,j}}[\Theta]} \cdot L_{\mu_{i,j}}^{-1} \le e_{i,j} \notag
\end{equation}
where $p_{i,j}, e_{i,j}$ are the adjusted confidence and the relative error 
requirement, respectively. The overall constraint $\Phi(\Theta)$ is defined as 
the conjunction of all individual $\phi_{i,j}(\Theta)$.

However, the sampling plan space defined by $\Phi(\Theta)$ is too broad to 
locate the most efficient sampling plan quickly. To further narrow down the plan 
space, we introduce the following additional conditions. First, due to the 
overhead of sampling, a query with a sampling rate larger than 10\% can be as 
expensive as the exact query (Figure \ref{fig:motivate-block}). Thus, we only 
consider sampling plans with sampling rates smaller than 10\%, which is 
consistent with prior work \cite{quickr}. Second, we only consider sampling plans 
that minimize the sample rate of one of the tables. Finally, we only sample 
large tables that are expensive to load, using a similar approach to how 
we identify tables to sample in the pilot query. We choose tables that 
will be scanned (not seeked) by the DBMS and are of high cardinality (\eg, fact tables 
\cite{datawarehouse}). In our experiment, we set a threshold of 1 million rows. 
These constraints result in the following space of sampling plans for a query with 
$l$ large tables.
\begin{equation*}
    \tilde{\Theta} := \bigl\{\arg\min_{\Theta} \theta_i,\ s.t.\ \Phi(\Theta) 
    \wedge D(\Theta, S)\ \big|\ S \subset \{1, \ldots, l\}, i \in S  \bigr\}
\end{equation*}
where $D(\Theta, S)$ defines the domain of sampling plans:
\begin{equation}
D(\Theta, S) :\equiv\quad 
    \left(\forall_{i \in S}\ 0 < \theta_i \le 0.1\right) \wedge 
    \left(\forall_{i \notin S}\ \theta_i = 1\right)
    \notag
\end{equation}

In \name, we enumerate the sets of tables to sample and the individual table 
of which we aim to minimize the sampling rate. For each optimization problem, 
we use the trust region method for fast and robust convergence \cite{trust-region}.

\minihead{Cost-based Optimization}
The solved sampling plans $\tilde{\Theta}$ often contain more than one plan.
Among them, we must choose the most efficient one to execute. Unfortunately, 
measuring the exact cost is prohibitively expensive, as it requires executing the
plan. Furthermore, cost estimation is a challenging problem, lacking a universal
solution for all DBMSs \cite{wu2013predicting}. In \name, we use the cost model 
of the underlying DBMS to estimate the cost. Most DBMSs offer external APIs to 
quickly estimate the cost of a query without executing it \cite{postgres-cost,
sqlserver-cost,presto-cost,impala-cost}. For in-memory databases that may not 
have cost estimators, such as DuckDB \cite{duckdb}, we estimate the cost by the 
volume of scanned data. This is because data scanning can be much more expensive 
than data processing for in-memory databases \cite{duckdb}. Empirically, the 
latency to sampling plan optimization is negligible compared to the overall 
query execution (\S \ref{sec:cost}).

Furthermore, exact queries are likely to be cheaper to execute than approximate queries 
with large sampling rates, particularly when small errors are required for 
queries with high selectivity or large group cardinality. To address it, \name 
rejects inefficient sampling plans when the estimated cost is larger than that of 
the exact query. If no sampling plan is feasible, \name will execute the exact 
queries.

\subsection{Query Rewriting} \label{subsec:rewrite}
Throughout \samplingscheme, we use query rewriting to synthesize and execute 
intermediate queries on the underlying DBMS. We describe the high-level 
procedures to rewrite an arbitrary aggregation query into (1) a pilot query 
$\pilotQ$ which computes statistics required by Procedure \ref{theorem:rel-error} 
and (2) a final query $\finalQ$ which computes the final answer based on the 
sampling plan optimized in Section \ref{subsec:optimize}. We demonstrate the 
query rewriting with an example in Figure \ref{fig:rewriting-example}.

\begin{figure}[t]
    \centering
    \includegraphics[width=\linewidth]{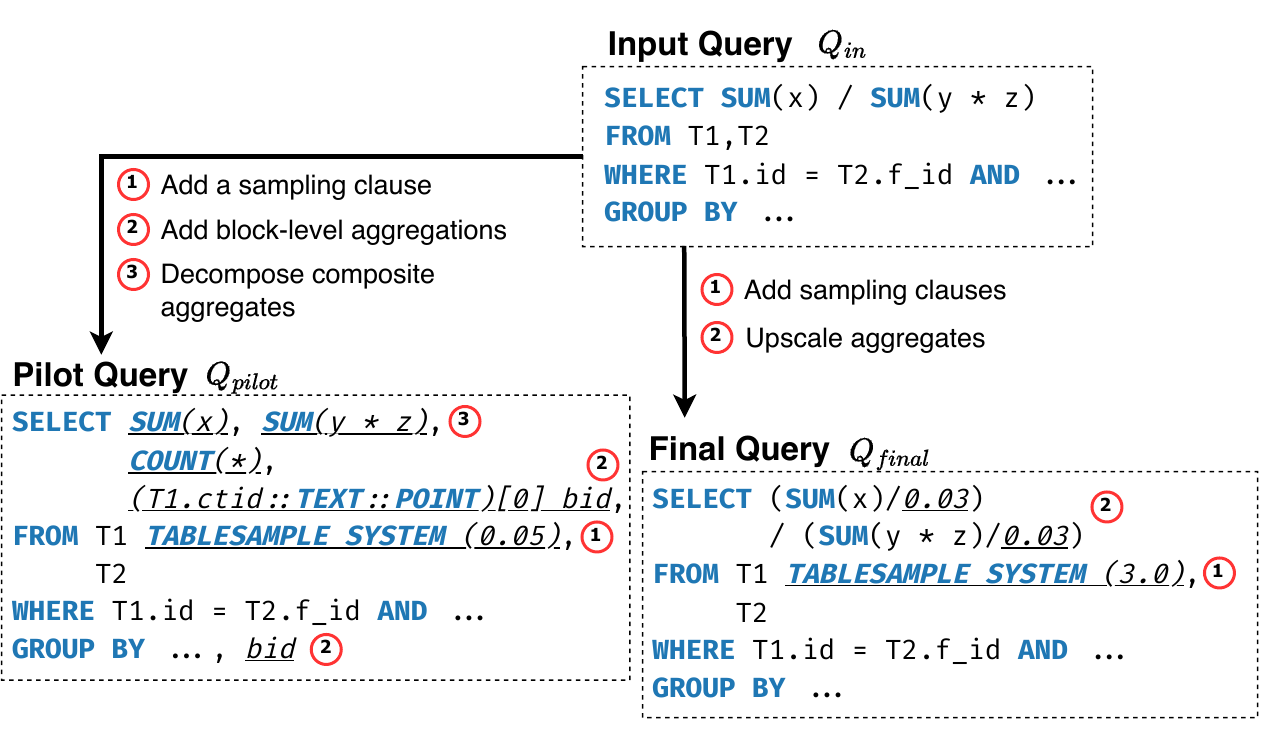}
    \caption{Demonstration of query rewriting with PostgreSQL syntax. Rewritten parts are emphasized.}
    \label{fig:rewriting-example}
\end{figure}

\minihead{Pilot Query Rewriting}
Based on Procedure \ref{theorem:rel-error}, $\pilotQ$ computes different 
statistics for different sampling methods. For row-level Bernoulli sampling, 
$\pilotQ$ can directly compute aggregates, corresponding standard deviations, 
and the sample size. For block sampling, $\pilotQ$ needs to calculate the 
aggregates and the size for each sampled block. This requires $\pilotQ$ to group 
the result by blocks. We achieve this by specifying the location of physical 
data blocks as a column expression.\footnote[5]{Nearly every DBMS that 
implements \texttt{TABLESAMPLE SYSTEM} supports outputing data location in some 
form \cite{psql-pageloc,sqlserver-pageloc,duckdb-pageloc,impala-pageloc,
presto-pageloc}.}
For example, in DuckDB, we divide the row ID by the block size; in PostgreSQL, 
we use the system column \texttt{ctid}. We summarize the rewriting procedures as 
follows:
\begin{enumerate}[leftmargin=*]
  \item We add a sampling clause (\eg, \texttt{TABLESAMPLE SYSTEM}) to the largest 
        table in $\inQ$.
  \item We incorporate the block location column of the largest table into 
        Group By clauses to compute block-level aggregates.
  \item We decompose composite aggregates (\eg, \texttt{SUM(x)/SUM(y)}) into 
        simple aggregates.
\end{enumerate}

\minihead{Final Query Rewriting}
The final query $\finalQ$ computes estimates of aggregates using the optimized 
sampling plan obtained. We summarize the rewriting procedures as follows:
\begin{enumerate}[leftmargin=*]
    \item We add sampling clauses according to the sampling plan.
    \item We upscale the \texttt{SUM}-like aggregates by dividing the product of 
          sampling rates.
\end{enumerate}

\section{Block Sampling for Efficient Online AQP} \label{sec:block}
In this section, we address the statistical challenges mentioned in Section 
\ref{subsec:bg-clg}. We first present motivations for using block sampling, 
examining its benefits and feasibility (\S \ref{subsec:block-motiv}). Next, we 
develop theoretical results that enable block sampling in AQP with statistical 
guarantees. That is, we obtain estimations required by Procedure 
\ref{theorem:rel-error} (\ie, $L_\mu$ and $U_V[\Theta]$) using block sampling 
for complex queries with nested subqueries (\S \ref{subsec:nested}) and Join 
(\S \ref{subsec:join}).

\subsection{Motivations}\label{subsec:block-motiv}
Throughout the history of AQP research, a wide range of sampling methods have 
been studied, but there is no universal best method \cite{chaudhuri2017approximate}.
Nevertheless, to simultaneously achieve \aPrioriError, \noMaintenance, and 
\notModifyDbms, we argue that block sampling, which samples data blocks,
is better than row-level sampling methods. We will explain this from three 
perspectives that are crucial in choosing sampling methods:
\begin{enumerate}[leftmargin=*]
    \item System Efficiency: volume of resulting data in a fixed time
    \item Statistical Efficiency: required sample size for a fixed error rate
    \item Feasibility: achieving statistical guarantees on various DBMSs
\end{enumerate}

\begin{figure}[t]
    \centering
    \includegraphics[width=\linewidth]{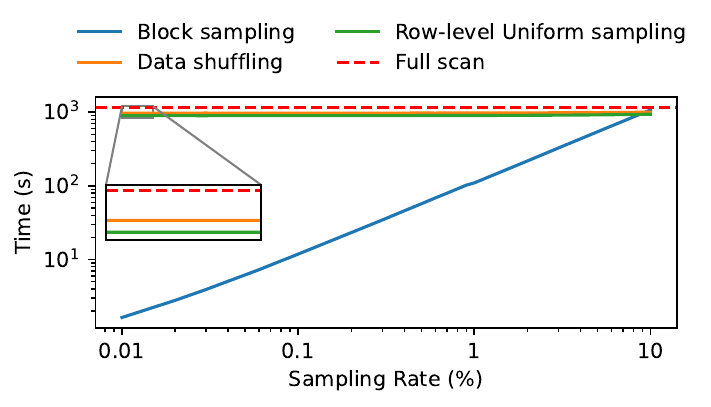}
    \caption{Comparison of the system efficiency of sampling methods that do not modify DBMSs. 
    At small sampling rates, such as 0.01\%, block sampling can be 500$\times$ faster than others.}
    \label{fig:motivate-block}
\end{figure}

\minihead{System Efficiency}
Across sampling methods that do not need DBMS modifications, block sampling 
achieves higher system efficiency than others. Block sampling improves system 
efficiency by skipping scanning non-sampled data. We evaluated the throughput of 
block sampling, row-level uniform sampling, and data shuffling on a 6B-row table. 
Figure \ref{fig:motivate-block} shows the latency to complete an \aavg query over 
the sampled data with sampling rates from 0.01\% to 10\% on PostgreSQL. At small 
sampling rates (\eg, 0.01\%), block sampling outperforms others by up to 
500$\times$. At large sampling rates (\eg, 10\%), all methods have comparable 
latencies to a full scan.

\minihead{Statistical Efficiency}
Block sampling can achieve comparable or higher statistical efficiency compared 
to row-level uniform sampling. Intuitively, block sampling introduces correlation 
across data from the same block, which seems to affect its statistical efficiency. 
However, in the case when the data of blocks is heterogeneous, the statistical 
efficiency of block sampling can be similar to or better than row-level uniform 
sampling. We analyze this with an \aavg query over a table 
$\{X_i|1 \le i \le N\cdot b\}$ of $N$ blocks and a consistent block size 
$b$.\footnote[6]{The analysis based on varied block sizes can be similarly derived 
by treating the block size as a random variable.} We present the theoretical 
result in Lemma \ref{lemma:comp-stats} and defer the proof to Appendix \ref{sec:app-proof}.

\begin{lemma} \label{lemma:comp-stats}
Let $\sigma^2_j$ be the variance of data in the $j$-th block. The ratio between 
the sample size of block sampling and that of row-level uniform sampling to 
achieve the same accuracy in expectation is 
$b\left(1 - \mathbb{E}\left[\sigma_j^2\right]\bigl/Var\left[X_i\right]\right)$.
\end{lemma}

Based on Lemma \ref{lemma:comp-stats}, we analyze the statistical efficiency of 
block sampling in two cases. First, when each data block is heterogeneous (\ie, 
$\mathbb{E}\left[\sigma_i^2\right] \rightarrow Var[X_{i,j}]$), the required sample 
size for block sampling can be smaller than that of row-level uniform sampling,
achieving better statistical efficiency. This can happen when the underlying 
DBMS has large data blocks. Second, when each data block is homogeneous 
(\ie, $\mathbb{E}\left[\sigma_i^2\right] \rightarrow 0$), the required sample size 
for block sampling is up to $b$ times that of row-level uniform sampling. We 
found that this rarely happens, especially with deep queries or complex 
predicates, and is often offset by the system efficiency of block 
sampling.

\minihead{Feasibility}
Finally, we evaluate whether it is feasible to use block sampling to approximately 
process arbitrary aggregation queries. We identify two key criteria for this 
to happen. First, can we obtain unbiased estimations \cite{quickr}? It is easy 
to verify that estimations of linear aggregates using block sampling are unbiased. 
For example, the \asum aggregate can be approximated without bias by adding 
summations of data blocks divided by the sampling rate. Second, can we achieve 
statistical guarantees of errors \cite{chaudhuri2017approximate}? For queries 
computing aggregates directly on the output of block sampling, we can achieve
error guarantees by analyzing block-level statistics \cite{hou1991statistical,
haas1996selectivity,pansare2011online}. For example, we can obtain a confidence 
interval of the mean of the sum of each block with standard CLT. However, it is 
non-trivial to achieve error guarantees for deep nested queries and Join queries. 
We dedicate the rest of this section to resolving it.

\begin{figure*}
    \centering
    \includegraphics[width=\linewidth]{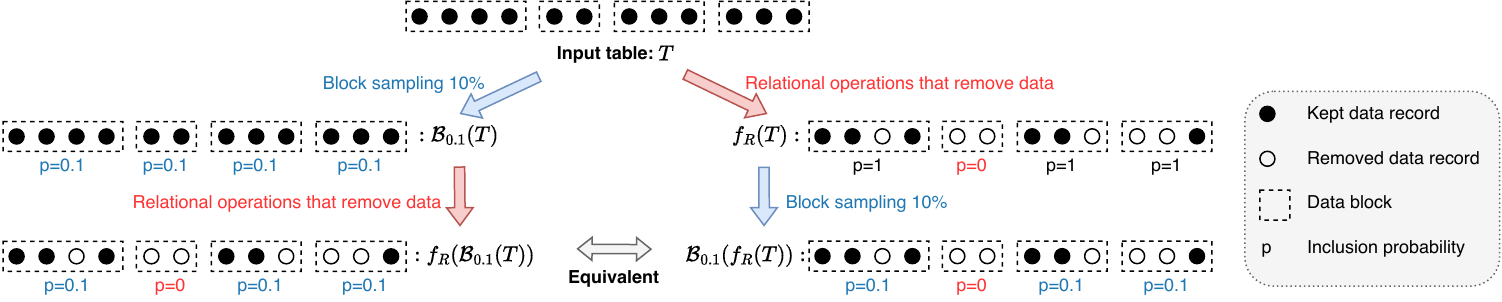}
    \caption{Demonstration of the commutativity between block sampling $\mathcal{B}$ and relational operations $f_R$ that remove data (\eg, \texttt{WHERE}, \texttt{JOIN} conditions, and \texttt{GROUP BY}); Order of operations does not affect the inclusion probability of each data block.}
    \label{fig:demo-equiv}
\end{figure*}
\subsection{Deep Nested Queries}\label{subsec:nested}
Achieving statistical guarantees for sampling-based AQP on deep nested queries 
is challenging, especially for non-uniform sampling methods \cite{quickr,sampling-algebra}, 
such as block sampling. This is because the output of sampling is manipulated by
subsequent relation operations, which potentially changes the statistical 
distribution of the sample. We use the following pair of queries as an example 
to demonstrate such a situation:
\begin{lstlisting}[keywords={SELECT,SUM,FROM,TABLESAMPLE,WHERE,DATE,AND,SYSTEM,JOIN,ON,AS}]
-- Q1: the query we execute
SELECT SUM(l_extendedprice * l_discount)
FROM lineitem TABLESAMPLE SYSTEM (0.5%) JOIN parts 
ON partkey 
WHERE l_shipdate >= DATE '1994-01-01' AND ...
-- Q2: the query we analyze
SELECT SUM(l_extendedprice * l_discount)
FROM ( SELECT * FROM lineitem JOIN parts ON partkey
       WHERE l_shipdate >= DATE '1994-01-01' AND ... ) 
AS cte TABLESAMPLE SYSTEM (0.5%)
\end{lstlisting}

We can obtain the confidence interval for Q2 by treating the sum of each block 
as a random variable, similar to prior work \cite{hou1991statistical,
haas1996selectivity,pansare2011online}. However, it is unclear how to calculate 
the confidence interval for Q1 due to the Join and filters between block sampling 
and the aggregation. In this section, we address this issue by analyzing the 
interaction between block sampling and relational operations and establishing 
rules for sampling equivalence.

\minihead{Intuition}
In general, we prove that block sampling is commutative with most relational 
operations, including projection, selection, Join, Group By, and Union. 
In Figure \ref{fig:demo-equiv}, we demonstrate that exchanging block sampling 
with any relational operation that removes data does not affect the probability 
distribution of the sample. For relational operations that add data (\eg, Join), 
we can always associate added data with a data block where block sampling operates.

\minihead{Formalization}
To formalize and prove this intuition, we define the notion of sampling equivalence 
in terms of sampling probability.
\begin{definition}\label{def:commutativity}
Two sampling procedures, $\mathcal{S}_1$ and $\mathcal{S}_2$, for a set of $k$ 
relations $\{T_1, \ldots, T_k\}$, where $k \ge 1$, are said to be equivalent, 
denoted as
\begin{equation}
    \mathcal{S}_1\left(\{T_1, \ldots, T_k\}\right) \Leftrightarrow 
    \mathcal{S}_2\left(\{T_1, \ldots, T_k\}\right) \notag
\end{equation}
if, for any possible sample result $R$, the probability of obtaining $R$ is the 
same under both sampling procedures $S_1$ and $S_2$, \ie,
\begin{equation}
    \forall R,\ \mathbb{P}\left[\mathcal{S}_1\left(\{T_1, \ldots, T_k\}\right)=R\right] 
    = \mathbb{P}\left[\mathcal{S}_2\left(\{T_1, \ldots, T_k\}\right)=R\right]. \notag
\end{equation}
\end{definition}
Next, we derive an important property of the sampling equivalence: the identity 
of the probability distribution of aggregates, as shown in the following 
proposition. We present the proof in Appendix \ref{sec:app-proof}.
\begin{proposition}\label{prop:identical-agg}
Let $\mathcal{S}_1$ and $\mathcal{S}_2$ be two equivalent sampling procedures. 
For any aggregate function $f$ that maps a table to a real value, the probability 
distribution of the $f$ applied to samples from $\mathcal{S}_1$ is identical to 
the probability distribution of $f$ applied to the samples from $\mathcal{S}_2$. 
Namely, for any real value $x$,
\begin{equation}
    \mathbb{P}\left[f\bigl(\mathcal{S}_1(\{T_1, \ldots, T_k\})\bigr) = x\right] 
    = \mathbb{P}\left[f\bigl(\mathcal{S}_2(\{T_1, \ldots, T_k\})\bigr) = x\right]\notag
\end{equation}
\end{proposition}

Based on Proposition \ref{prop:identical-agg}, to show the aggregates computed 
over the outputs of two different sampling procedures have the same distribution, 
it is sufficient to prove two sampling procedures are equivalent. Leveraging 
this, we show that block sampling is commutative with selection, Join, and Union 
in the following propositions. We present the proof in Appendix 
\ref{sec:app-proof}.

\begin{proposition}\label{theorem:select-comm}
(\textsc{Selection}) For any table $T$, selection $\sigma_\sqlpredicate$ with a 
predicate $\sqlpredicate$, and block sampling $\mathcal{B}_\theta$ with a 
sampling rate $\theta$,
\begin{equation}
    \sigma_\sqlpredicate(\mathcal{B}_\theta(T)) \Leftrightarrow \mathcal{B}_\theta(\sigma_\sqlpredicate(T)) \notag
\end{equation}
\end{proposition}

\begin{proposition}\label{theorem:join-comm}
(\textsc{Join}) For any tables $T_1$ and $T_2$, Join $\Join_\sqlpredicate$ with 
a predicate $\sqlpredicate$, and block sampling $\mathcal{B}_\theta$ with a 
sampling rate $\theta$,
\begin{equation}
    \mathcal{B}_\theta(T_1) \Join_\sqlpredicate T_2 \Leftrightarrow \mathcal{B}_\theta(T_1 \Join_\sqlpredicate T_2) \notag
\end{equation}
\end{proposition}

\begin{proposition}\label{theorem:union-comm}
(\textsc{Union}) Let $\cup$ be a bag union operation (or \texttt{UNION ALL} in 
SQL). For any tables $T_1, \ldots, T_k$ ($k \ge 2$) and block sampling 
$\mathcal{B}_\theta$ with a sampling rate $\theta$,
\begin{equation}
    \bigcup_{i=1}^k \mathcal{B}_\theta(T_i) \Leftrightarrow \mathcal{B}_\theta\left(\bigcup_{i=1}^k T_i\right) \notag
\end{equation}
\end{proposition}

Finally, we consider projection and Group By. We find that the commutativity 
between block sampling and project is trivial, since projection is at the column 
level and thus orthogonal to sampling. Moreover, Group By operations can be 
considered as a special case of selection with a predicate on the grouping columns.

We conclude these equivalence rules with the following standard form 
for any supported aggregation query $Q$:
\begin{equation}
    Q \Leftrightarrow \texttt{AGG}\left(\Join_{i=1}^k \mathcal{B}_{\theta_i} (\tilde{T}_i)\right)
    \label{eq:s-form}
\end{equation}
where $\tilde{T}_i$ is the output table of intermediate relational operations and
$\theta_i$ is the sampling rate of the $i$-th input table. This result is obtained 
by applying our equivalence rules recursively across the query. Intuitively, if 
an aggregation query executes block sampling on one input table ($k=1$), it is 
equivalent to the query that computes aggregate directly on a block sample. In
this case, we can calculate the estimates at the block level and use standard 
techniques to analyze the error \cite{haas1996selectivity}. If a query executes 
block sampling on multiple input tables ($k>1$), it is equivalent to the query 
that computes aggregate on the Join of block samples. 

We show that our sampling equivalence rules are stronger than 
sampling dominance rules of \textsc{QuickR}. First, the sampling dominance rules 
ensure accuracy dominance in only one direction and do not establish the 
equivalence. Second, using dominance rules are insufficient for proving the 
equivalence, as they only consider the inclusion probability of one or two 
sampled units (i.e., c- and v-dominance). In contrast, our equivalence rules 
consider the joint inclusion probability of the entire sample. As a result, when 
two sampling plans are equivalent in our definition, they inherently satisfy
sampling dominance. 

\subsection{Join Queries} \label{subsec:join}
When the input query has multiple large tables, \samplingscheme tries to
execute block sampling on multiple tables, which leads to Equation \ref{eq:s-form} 
with $k>1$. To analyze the query error with \samplingscheme, we need
to (1) ensure Procedure \ref{theorem:rel-error} is valid by investigating the 
asymptotic distribution of the aggregate over the Join of multiple block samples 
and (2) obtain two estimates $L_\mu$ and $U_V[\Theta]$ that are necessary
for \samplingscheme to plan sampling (\S \ref{subsec:plan}).

\minihead{Failure of the Naive Method}
However, due to correlations within blocks and across Join results, the asymptotic
distribution of Equation \ref{eq:s-form} with $k>1$ is not governed by standard 
CLT \cite{huang2019joins,chaudhuri1999random,haas1996selectivity}. Naively applying
the standard CLT to calculate confidence intervals can lead to invalid guarantees.
We show this failure through the following query that Joins two large tables and 
uses block sampling on both tables: 
\begin{lstlisting}[keywords={SELECT,COUNT,FROM,TABLESAMPLE,WHERE,LIKE,SYSTEM,AND,
INNER,JOIN}]
SELECT SUM(price) FROM lineitem TABLESAMPLE SYSTEM(1%) 
INNER JOIN orders TABLESAMPLE SYSTEM(5%) 
WHERE l_orderkey = o_orderkey AND comment LIKE '%special%'
\end{lstlisting}
The ``confidence interval'' obtained through standard CLT with a 95\% intended 
confidence may only achieve a coverage probability as low as 8\%.\footnote[7]{We 
evaluated the query on DuckDB with the 1,000-scaled TPC-H 1,000 times.} 

\minihead{Our Solutions} 
We show that the sample mean still asymptotically converges to a normal 
distribution when multiple tables of a Join operation are sampled at the block 
level. However, the variance is not in the standard form. We first
present the asymptotic convergence in Theorem \ref{theorem:join-clt} and defer 
the proof to Appendix \ref{sec:app-proof}. Theorem \ref{theorem:join-clt} are inspired by 
\cite{haas1996selectivity} but extends their theory to sampling with different 
rates. We present the theorem in a standard way using the block-level 
\texttt{AVG} aggregate. The result for \asum and \acount can be obtained
similarly, while the row-level \aavg can be considered as a ratio between \asum 
and \acount. 
\begin{theorem}\label{theorem:join-clt}
Suppose a Join operation is executed on a set of $k$ tables $\{T_1, \ldots, T_k\}$,
where each table $T_i$ has a set of $N_i$ blocks: $\{t_{i, 1}, \ldots, t_{i, N_1}\}$.
Let $\myjoin(*)$ be a function that takes as input $k$ blocks of different 
tables and produces the sum of the Join result of these blocks. We denote $\mu$ as
the block-level mean of the Join result:
\begin{equation}
\mu = \left(\prod_{i=1}^k N_i\right)^{-1} 
      \sum_{i_1=1}^{N_1}\dots\sum_{i_k=1}^{N_k} 
      \mathcal{J}(t_{1, i_1}, \ldots, t_{k, i_k})
\end{equation}
For each Join table $T_i$, we execute the block sampling with a sample size of
$n_i$ blocks. We denote $\hat \mu$ as the block-level mean of the Join result of 
block samples. Then, we can have the following convergence
\begin{equation}
    \hat\mu - \mu \xrightarrow{D} \mathcal{N}(0, Var[\hat\mu]) \quad
    \text{as} \quad
    n_i \rightarrow \infty
\end{equation}
where $Var[\hat\mu]$ is the (unknown) variance of $\hat\mu$.
\end{theorem}

Theorem \ref{theorem:join-clt} validates our \samplingscheme algorithm on queries 
where multiple tables are sampled at the block level. To obtain concrete 
sampling plans, Procedure \ref{theorem:rel-error} requires a lower bound of 
aggregate: $L_\mu$ and an upper bound of the variance of the aggregate estimator: 
$U_V[\Theta]$. We show the results of $U_V[\Theta]$ for the two-table sampling 
with a \asum aggregate.
$L_\mu$ can be derived based on standard probabilistic inequalities, such as 
Chebyshev's Inequality \cite{intro-math-stats}. We defer the proof to Appendix
\ref{sec:app-proof}.
\begin{lemma}\label{lemma:var-ub}
Consider a query which Joins two tables $T_1$ and $T_2$. Without loss of generality,
we suppose that in the pilot query, block sampling with a tiny sampling rate 
$\theta_p$ is executed on $T_1$, resulting in $n_p$ blocks. Given a final sampling plan 
$\Theta=[\theta_1, \theta_2]$, the probability that the variance of the \asum 
estimate has an upper bound defined as follows is at least $1-\delta_2$:
\begin{align*}
U_V[\Theta] 
&= \frac{1-\theta_1}{\theta_1}U_{y^{(1)}}\left[\frac{\delta_2}{N_2+2}\right]
+ \frac{1-\theta_2}{\theta_2} \sum_{i_2=1}^{N_2} \left(U_{y_{i_2}^{(2)}}\left[\frac{\delta_2}{N_2+2}\right] \right)^2 \\
&+\frac{(1-\theta_1)(1-\theta_2)}{\theta_1\theta_2}U_{y^{(3)}}\left[\frac{\delta_2}{N_2+2}\right]
\end{align*}
where $y^{(1)}_i = \left(\sum_{i_2=1}^{N_2} \mathcal{J}\left(t_{1, i}, t_{2, i_2}\right)\right)^2$, 
$y^{(2)}_{i_2,i} = \mathcal{J}\left(t_{1, i}, t_{2, i_2}\right)$, 
$y^{(3)}=\sum_{i_2=1}^{N_2} \mathcal{J}\left(t_{1, i}, t_{2, i_2}\right)^2$, 
and $U_y\left[\delta\right]$ is the upper bound of the Student's t confidence interval of the summation of $y$ 
with $1-\delta$ confidence \cite{intro-math-stats}.
\end{lemma}

% \subsection{Query Rewriting}

\section{Evaluation}\label{sec:eval}
In this section, we evaluate \name with experiments to answer the 
following questions:
\begin{enumerate}[leftmargin=*]
    \item Does \name achieve statistical guarantees (\S \ref{subsec:eval-correct})? 
    % Our \samplingscheme algorithm guarantees user-specified errors.
    \item How much can \name accelerate queries (\S \ref{sec:perf})? 
    % We compare against exact query processing on various DBMSs and a 
    % state-of-the-art online AQP system: \textsc{Quickr} \cite{quickr}.
    \item How much can \bstats improve existing online AQP (\S \ref{subsec:eval-augment})? 
    % We augment \textsc{Quickr} by replacing row-level uniform sampling with \bstats 
    \item What are the individual contributions of \samplingscheme and \bstats 
    to overall performance (\S \ref{subsec:ablation})?
    % \item How does \name perform in different settings?
\end{enumerate}

\subsection{Experiment Settings}

\begin{table}
    \centering
    \caption{Characteristics of workloads.}
    \label{tab:workloads}
    \begin{tabular}{cccc}
        \toprule
        Benchmark    & \#Queries & \#Queries w/ Join & Max/Avg. \#groups \\
        \midrule
        TPC-H       & 9         & 7                          & 175/22            \\
        ClickBench  & 7         & 0                          & 17/3            \\
        SSB & 10        & 10                         & 150/38            \\
        Instacart   & 9         & 7                          & 146/22           \\
        DSB-DBest       & 169       & 42                         & 261/52           \\
        \bottomrule
    \end{tabular}
\end{table}
\minihead{Benchmarks} 
We evaluate \name on a diverse set of benchmarks, including four benchmarks that 
are widely used in prior work \cite{blinkdb,verdictdb,sample+seek,clickbench-ref1,
ssb-ref1,ssb-ref2,baq} and a benchmark that simulates real-world data with skewed 
distributions \cite{ding2021dsb}. Other real-world benchmarks used in prior work 
are proprietary \cite{sample+seek,blinkdb}. Thus, we cannot evaluate \name on 
those benchmarks.
\begin{itemize}[leftmargin=*]
    \item \textbf{TPC-H} and \textbf{SSB} are synthetic benchmarks for 
    decision-making \cite{tpch} and star-schema data warehousing 
    \cite{ssb-paper}, respectively. We use a scale factor of 1,000.
    \item \textbf{ClickBench} is a real-world benchmark obtained from the 
    traffic recording of web analytics \cite{clickbench}. We scale up the raw 
    data by $5\times$, resulting in a pre-processed size of 200GB.
    \item \textbf{Instacart} is a micro-benchmark with real-world data from the 
    Instacart \cite{instacart} and queries from TPC-H. We scale up the original 
    data by 100$\times$ using the same method as \textsc{VerdictDB} 
    \cite{verdictdb}.
    \item \textbf{DSB} is a synthetic benchmark based on TPC-DS, blended with 
    skewed yet real-world data distributions, including the (bucketed) 
    exponential distribution and correlations across columns \cite{ding2021dsb}. 
    We use a scale factor of 1,000. To cover the skewness in aggregation, Join, 
    and Group By columns, we use the queries from 
    \textsc{DBest} \cite{dbest}.
\end{itemize}
In line with previous AQP studies \cite{verdictdb,taster}, we exclude 
queries with an empty result, correlated subqueries, and a large group cardinality. 
In production scenarios, \name can identify those queries via \samplingscheme 
and execute the exact query. We summarize the key statistics of the workloads in 
Table \ref{tab:workloads}. A large portion of queries contain Join and various 
numbers of groups.

\minihead{DBMSs}
We evaluate \name on three DBMSs: PostgreSQL 16.3, SQL Server 2022, and DuckDB 
1.0. DuckDB is an open-source in-memory column-oriented DBMS\cite{duckdb}. The 
default block sampler of DuckDB always scans the entire column, which is less 
efficient compared to PostgreSQL and SQL Server. To improve the efficiency of 
DuckDB's block sampling, we add optimization rules that push down block sampling 
to sequential scanning. Our extension has been merged in DuckDB 1.2.

\minihead{Baselines}
As far as we know, \name is the first AQP system that simultaneously achieves 
\aPrioriError, \noMaintenance, and \notModifyDbms. There are no directly 
comparable AQP systems to use as a baseline. Hence, we compare \name with 
executing exact queries on DBMSs that have state-of-the-art query optimizations. 
In addition, we compare with \textsc{Quickr} \cite{quickr}, the state-of-the-art 
online AQP system. \textsc{Quickr} achieves \aPrioriError and \noMaintenance 
but fails to fulfill \notModifyDbms, which is the closest to \name.

\minihead{Testbed}
Our experiments are conducted on CloudLab \cite{cloudlab} r6525 nodes, each 
equipped with 256GB RAM, 1.6TB NVMe SSD, and two 32-core AMD 7543 CPUs.
\footnote[8]{256GB RAM is large enough for DuckDB to fit in required columns for 
individual queries after default compressions.} Before executing each query, we 
clear both the operating system cache and the query plan cache.

\begin{figure*}[t!]
    \begin{minipage}{\linewidth}
        \centering
        \begin{subfigure}[t]{0.24\linewidth}
            \centering
            \includegraphics[width=\linewidth]{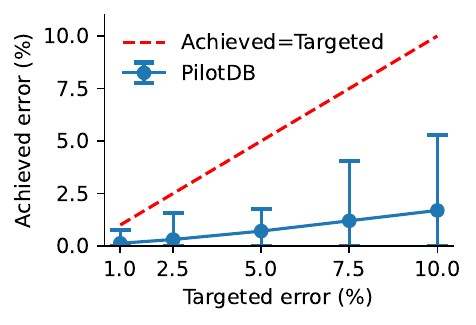}
            \caption{TPC-H.}
            \label{fig:tpch-error}
        \end{subfigure}%
        \hfill
        \begin{subfigure}[t]{0.24\linewidth}
            \centering
            \includegraphics[width=\linewidth]{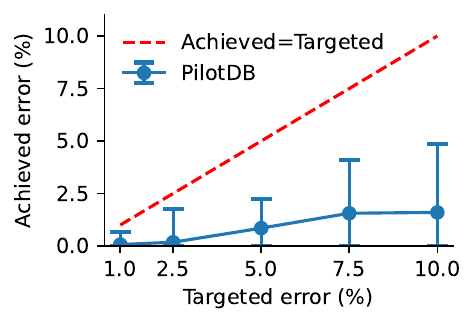}
            \caption{SSB.}
            \label{fig:ssb-error}
        \end{subfigure}%
        \hfill
        \begin{subfigure}[t]{0.24\linewidth}
            \centering
            \includegraphics[width=\linewidth]{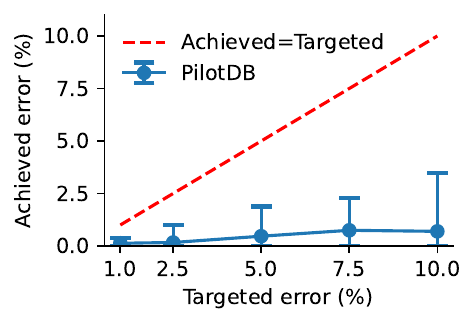}
            \caption{ClickBench.}
            \label{fig:clickbench-error}
        \end{subfigure}%
        \begin{subfigure}[t]{0.24\linewidth}
            \centering
            \includegraphics[width=\linewidth]{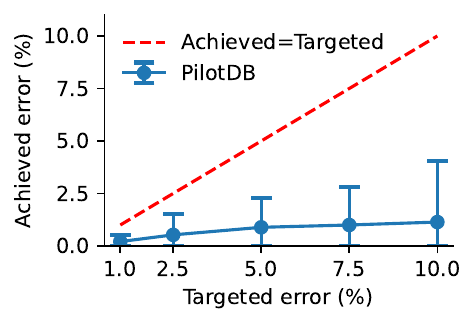}
            \caption{Instacart.}
            \label{fig:instacart-error}
        \end{subfigure}%
    \caption{\name achieves error guarantees on TPC-H, SSB, ClickBench, and 
    Instacart. The achieved error is smaller than targeted error if the result
    is below the red dashed line. We show the maximum, mean, and minimum errors 
    in 20 executions.}
    \label{fig:agg-error}
    \end{minipage}
\end{figure*}
\begin{figure}[t]
    \centering
    \includegraphics[width=0.6\linewidth]{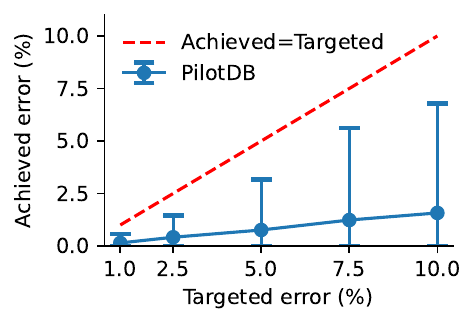}
    \caption{\name achieves the targeted errors on DSB, which has skewness on aggregation columns, Join columns, and Group By columns.}
    \label{fig:dsb-error}
\end{figure}

\subsection{\name Guarantees Errors} \label{subsec:eval-correct}
We first evaluated whether \name achieves a priori error guarantees. We executed 
each query from the five benchmarks on PostgreSQL 20 times, each with different 
targeted error rates--the maximum relative error in the specification (\S 
\ref{subsec:error-semantic}). We set the confidence to 95\% and measured the 
maximum relative error of aggregates. By default, we sampled at 0.05\% during the 
planning stage of \samplingscheme. If the input query has Group By clauses, we 
use Lemma \ref{lemma:group-error} with $g=200, p_f=0.05$ to compute the sampling 
rate for planning.

Figures \ref{fig:agg-error} and \ref{fig:dsb-error} show the achieved errors for 
each benchmark with various targeted errors. The bars in the figure represent the
minimum and maximum achieved errors across all queries and executions, while the 
dots indicate the average achieved errors. For reference, we plot a dashed red 
line to show the case when the achieved error equals the targeted error. As 
shown, the achieved errors of \name are always less than the targeted errors. 
Furthermore, we find that none of the evaluated queries miss groups.

We observe that \name guarantees errors conservatively, with the maximum achieved 
errors being approximately half of the targeted errors. This arises because the 
sampling rates determined by \samplingscheme are guaranteed to be sufficiently 
large but may not always be the minimum necessary to meet the user's error 
specifications. For example, we apply Boole's Inequality to tackle the joint 
probability of multiple events. The equality holds only when events are 
mutually exclusive. To ensure the sampling rates are also the minimum 
necessary, it is crucial to analyze the correlations between aggregates, which 
will be a future work.

We also evaluated the achieved errors when \bstats is replaced with a standard 
CLT-based confidence interval. We show that without \bstats, the achieved error
can be up to 52$\times$ higher (1.7$\times$ higher on average) than the target 
error, highlighting the contribution and necessity of \bstats. 

\subsection{\name Accelerates Query Processing} \label{sec:perf}
We analyze the performance of \name by evaluating it on various 
DBMSs, with different targeted errors, and across all five benchmarks. The query 
execution follows the setting in Section \ref{subsec:eval-correct}.

\minihead{\name Accelerates Queries across Various DBMSs}
We evaluated \name on TPC-H, ClickBench, SSB, and Instacart across three DBMSs, 
targeting a 5\% error and 95\% confidence. We executed each query in each DBMS 10 
times and calculated the geometric mean (GM) of speedups. 

Figure \ref{fig:performance} provides a detailed view of performance on each 
database, showing the cumulative probability function (CDF) of speedups compared 
to exact query execution. As shown, \name consistently accelerates 80\% of queries 
across all DBMSs. Moreover, \name achieves up to 126$\times$ speedup on transactional 
databases and up to 13$\times$ speedup on an analytical database, DuckDB. In the 
worst case, \name slows down the execution by at most 8\%. This is 
because the sample planning stage involve executing a pilot query, the primary 
overhead causing the loss in performance.

We observe that \name performs better on PostgreSQL and SQL Server 
than on DuckDB. This is because DuckDB is optimized for in-memory processing. 
When the data fits in the memory, DuckDB processes queries faster than 
transactional databases.
\begin{figure*}[t!]
    \begin{minipage}{0.74\linewidth}
        \begin{subfigure}{0.33\linewidth}
            \centering
            \includegraphics[width=\linewidth]{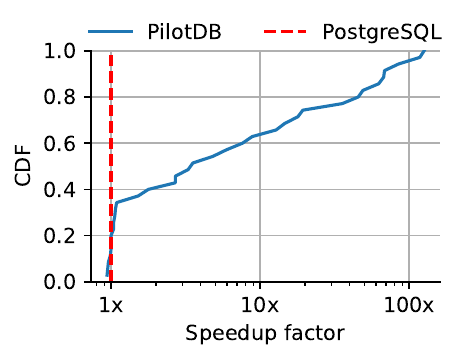}
            \caption{PostgreSQL (log scale).}
            \label{fig:postgres}
        \end{subfigure}%
        \hfill
        \begin{subfigure}{0.33\linewidth}
            \centering
            \includegraphics[width=\linewidth]{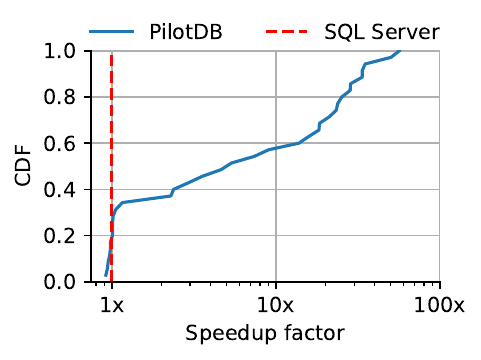}
            \caption{SQL Server (log scale).}
            \label{fig:sqlserver}
        \end{subfigure}%
        \hfill
        \begin{subfigure}{0.33\linewidth}
            \centering
            \includegraphics[width=\linewidth]{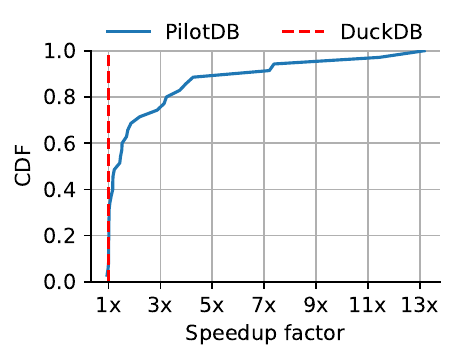}
            \caption{DuckDB.}
            \label{fig:duckdb}
        \end{subfigure}
        \caption{\name achieves 0.92-126$\times$ speedups over exact 
        execution across three DBMSs.}
        \label{fig:performance}
    \end{minipage}
    \begin{minipage}{0.24\linewidth}
        \includegraphics[width=\linewidth]{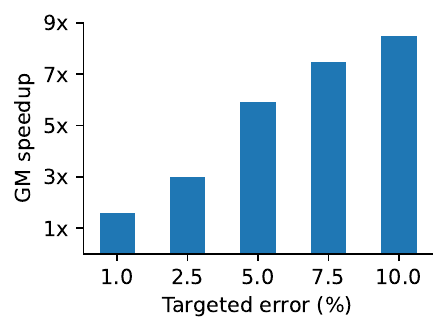}
        \caption{Speedups of \name across various errors.}
        \label{fig:error-perf}
    \end{minipage}
\end{figure*}

\begin{figure}[t!]
    \begin{subfigure}[t]{0.48\linewidth}
        \centering
        \includegraphics[width=\linewidth]{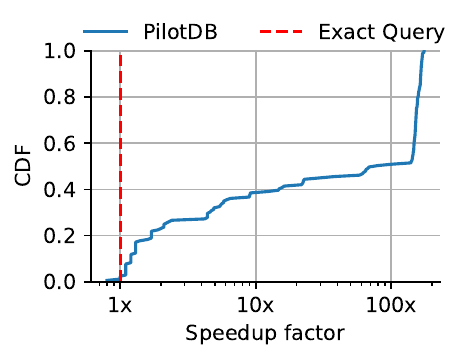}
        \caption{Detailed performance on individual queries.}
        \label{fig:skew-cdf}
    \end{subfigure}%
    \hfill
    \begin{subfigure}[t]{0.48\linewidth}
        \centering
        \includegraphics[width=\linewidth]{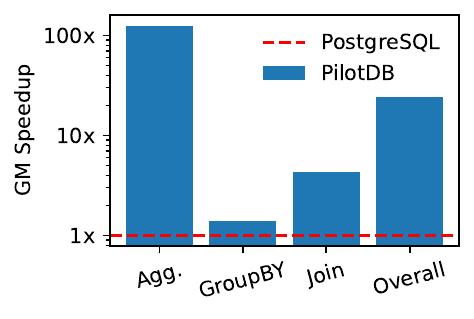}
        \caption{Performance grouped by query types.}
        \label{fig:skew-types}
    \end{subfigure}
    \caption{\name accelerates queries on skewed data.}
\end{figure}

\minihead{\name Accelerates Queries on Skewed Data}
To demonstrate the performance of \name on skewed data distributions, we evaluated 
\name on DSB with a diverse set of 97 aggregation queries, 30 queries with 
Group By, and 42 queries with Join \cite{dbest}. We executed each query 10 times 
on PostgreSQL and calculated the geometric mean of speedups.

Figure \ref{fig:skew-cdf} shows the CDF of query speedups of \name on DSB. As 
shown, \name accelerates queries over skewed data by up to two orders of 
magnitude compared to exact queries on PostgreSQL. To understand how \name 
performs on different types of queries and skewness, we group query speedups by 
the query type in Figure \ref{fig:skew-types}. ``Agg.'' refers to simple 
aggregation queries where the data of aggregated columns is exponentially 
distributed. ``GroupBy'' and ``Join'' refer to queries with exponentially 
distributed data in the Group By dimension or Join columns, respectively. \name 
achieves 55$\times$ overall speedup and 125$\times$ speedup on simple 
aggregation queries. On Group By and Join queries, \name achieves 1.4$\times$ and 
4.3$\times$ speedup, respectively. This is relatively small compared to simple 
aggregation queries, but still significant compared to row-level uniform 
sampling which has 0.9$\times$ speedup on average.

\minihead{\name Accelerates Queries with Various Error Targets}
To study how \name performs with different error targets, we evaluated the 
performance of \name with error targets 1\%-10\% on PostgreSQL. We executed 
each query 10 times for each error target and calculated the geometric mean of 
speedups.

Figure \ref{fig:error-perf} shows the speedup according to different targeted 
errors. We observe that \name achieves query speedups for all evaluated targeted 
errors. Even with a small targeted error of 1\%, \name achieves 1.6$\times$ 
speedup. As expected, we find that \name achieves higher speedups at larger 
targeted errors.

\minihead{Comparison with \textsc{Quickr}}
We compared \name with the state-of-the-art online AQP system \textsc{Quickr}.
Since \textsc{Quickr} is not open-sourced, we consider a strict performance 
upper bound of it. Specifically, as mentioned explicitly in their paper 
\cite{quickr}, \textsc{Quickr} requires one pass over the data. Therefore, we 
consider the data scanning time on each DBMS as the performance upper bound 
(\ie, latency lower bound) of \textsc{Quickr}. We give \textsc{Quickr} the 
benefit of parallelizing scanning with all CPU cores and only consider the 
elapsed time of the longest scanning operation.

\begin{figure}
    \begin{minipage}{0.48\linewidth}
        \centering
        \includegraphics[width=\linewidth]{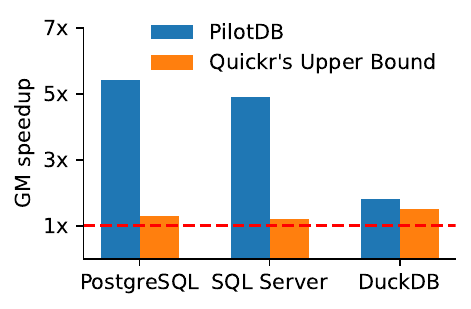}
        \caption{\name outperforms \textsc{Quickr} by up to 4.2$\times$ across
        three DBMSs.}
        \label{fig:quickr}
    \end{minipage}
    \hfill
    \begin{minipage}{0.48\linewidth}
        \centering
        \includegraphics[width=\linewidth]{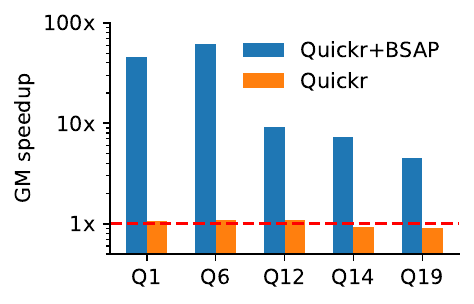}
        \caption{\bstats improves the speedup of \textsc{Quickr} by up to 
        60$\times$ on DuckDB.}
        \label{fig:quickr-aug}
    \end{minipage}
\end{figure}

Figure \ref{fig:quickr} demonstrates the upper bound speedup of \textsc{Quickr} 
and the speedup of \name across three DBMSs. As shown, \name demonstrates significantly 
higher query speedup by 1.2-4.2$\times$. Compared to \textsc{Quickr} 
which always scans the whole data, \name achieves better efficiency by skipping 
non-sampled data blocks.

\subsection{\bstats Augments Existing Online AQP} \label{subsec:eval-augment}
In this section, we evaluated whether and how much \bstats can improve the 
performance of existing online AQP. We used TPC-H queries where \textsc{Quickr} 
applies row-level uniform sampling. On those queries, we reproduce 
\textsc{Quickr} in DuckDB by manually adopting the rules described in 
\cite{quickr} and then rewriting queries with parallelized row-level uniform 
sampling. We incorporate \bstats into \textsc{Quickr} by further (1) replacing 
the uniform sampling with block sampling and (2) adapting the Horvitz-Tompson 
estimator with the error analysis of \bstats. Finally, we target a 10\% error, 
which is consistent with the setting in \textsc{Quickr}'s paper \cite{quickr}.

Figure \ref{fig:quickr-aug} shows the speedups of \textsc{Quickr}+\bstats and 
original \textsc{Quickr}, compared to exact queries on DuckDB. As shown, 
\textsc{Quickr}+ \bstats achieves higher speedups by 4.9-60$\times$. 
We find that these evaluated queries typically have a latency bottleneck at 
table scanning. In this case, \bstats can significantly accelerate existing 
online AQP by skipping non-sampled blocks when scanning tables.

\subsection{Ablation Study}\label{subsec:ablation}
We evaluated the effectiveness of the design choices of \name by 
comparing \name with its alternative configurations.
\begin{enumerate}[leftmargin=*]
    \item We replace \samplingscheme with pre-computed statistics (\oracle).
    \item We replace \bstats with row-level sampling (\uniform).
    \item We replace Bernoulli sampling with fixed-size sampling.
\end{enumerate}
We used the same setting as Section \ref{sec:perf} for query executions.
\begin{figure*}[t!]
    \centering
    \begin{minipage}{0.30\linewidth}
        \centering
        \includegraphics[width=\linewidth]{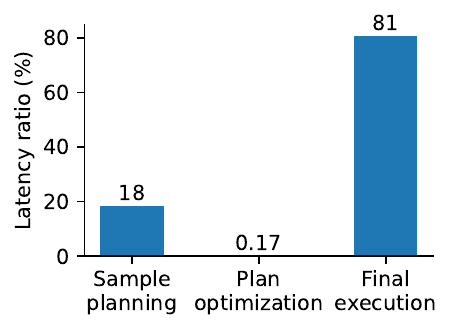}
        \caption{Latency decomposition of \name.}
        \label{fig:cost}
    \end{minipage}
    \hfill
    \begin{minipage}{0.33\linewidth}
        \centering
        \includegraphics[width=\linewidth]{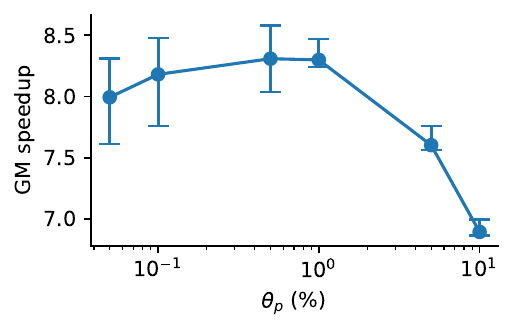}
        \caption{\name achieves >6$\times$ speedup across various $\theta_p$.}
        \label{fig:theta_p}
    \end{minipage}
    \hfill
    \begin{minipage}{0.32\linewidth}
        \centering
        \includegraphics[width=\linewidth]{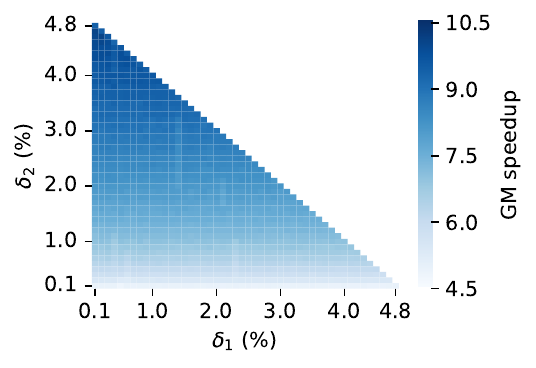}
        \caption{\name achieves 4.8-10.0$\times$ speedups for various $(\delta_1, \delta_2)$.}
        \label{fig:deltas}
    \end{minipage}
\end{figure*}

\minihead{\name Achieves Near-Optimal Performance}
In \samplingscheme, we use estimations based on a pilot query to determine the 
sampling rates for a given error specification (\S \ref{subsec:error-semantic}). 
To understand the impact of those estimations on the performance of \name, we 
compare it with \oracle, which represents the upper-bound performance achievable 
for AQP that uses online block sampling. For each query, we measure the latency
of \oracle, \name, and the second stage of \name. We executed all queries 
in our benchmarks.
\begin{table}[t]
    \centering
    \caption{Geometric mean of the slowdowns of \name compared to \oracle.}
    \label{tab:oracle}
    \begin{tabular}{lccc}
        \toprule
             & PostgreSQL & SQL Server & DuckDB \\
        \midrule
        PilotDB (overall)     & 1.61$\times$ & 1.21$\times$ & 1.27$\times$ \\
        PilotDB (2nd stage)   & 1.04$\times$ & 1.08$\times$ & 1.19$\times$ \\
        \bottomrule
    \end{tabular}
\end{table}

Table \ref{tab:oracle} shows the slowdowns of \name compared to 
\oracle, computed as the ratio of their latencies. Compared to \oracle, \name is 
only 21\%-61\% slower, showing the effectiveness of \samplingscheme. However, 
the latency of \oracle does not include the time to determine sampling rates, 
which requires executing the original input query. To decouple factors that affect 
the final latency, we also exclude the time to determine the sampling rates in 
\name, leaving the latency of the second stage of \name. As shown in Table 
\ref{tab:oracle}, the latency of the second stage of \name is only 4\%-19\% 
higher than \oracle. This demonstrates that the optimized sampling plan of \name 
is close to optimal.

% \begin{figure}[t]
%     \centering
%     \includegraphics[width=0.8\linewidth]{figures/pilotdb_uniform.pdf}
%     \caption{\name improves uniform sampling by up to 219$\times$. ``GM'' means geometric mean.}
%     \label{fig:uniform}
% \end{figure}

\begin{table}[t]
    \centering
    \caption{Speedups of \name over \uniform.}
    \label{tab:uniform}
    \begin{tabular}{lccc}
        \toprule
             & PostgreSQL & SQL Server & DuckDB \\
        \midrule
        Geometric mean     & 12.6$\times$ & 9.37$\times$ & 1.92$\times$ \\
        Maximum            & 219$\times$  & 71.4$\times$  & 13.2$\times$\\
        \bottomrule
    \end{tabular}
\end{table}

\minihead{\name Outperforms Row-level Bernoulli Sampling}
In Section \ref{sec:block}, we showed the advantage of \bstats over uniform 
row-level sampling with a motivating experiment in Figure \ref{fig:motivate-block}. 
Here, we demonstrate the benefit of \bstats in terms of end-to-end latency. We 
compared \name and \uniform across all the benchmarks. In \uniform, we use the 
default row-level Bernoulli sampling in each DBMS as the sampling method. That 
is, we rewrite queries in PostgreSQL and DuckDB with 
``\texttt{TABLESAMPLE BERNOULLI(p)}'', where ``\texttt{p}'' is the sampling rate, and in 
SQL Server with ``\texttt{WHERE rand() < p}'', where ``\texttt{rand()}'' outputs a 
random number in [0,1].

Table \ref{tab:uniform} summarizes the speedup of \name compared to \uniform.
We show the geometric mean and maximum speedup for each DBMS. \name achieves a 
higher geometric mean speedup by 8.0$\times$ and a higher maximum speedup by 
219$\times$, compared to \uniform. We observe that \name provides a greater 
benefit on PostgreSQL and SQL Server compared to DuckDB. This is because DuckDB 
is columnar, which, unlike Postgres and SQL Server, allows it to scan selected columns.

\color{black}
\minihead{Comparison with Fixed-size Sampling}
We compare \name with fixed-size sampling at the row and block level. We use 
``\texttt{ORDER BY RANDOM() LIMIT {sample\_size}}'' for row-level fixed-size sampling.
Furthermore, only PostgreSQL supports block-level fixed-size 
sampling, via an extension: \texttt{tsm\_system\_rows} \cite{postgres-block-fix-size-sampling}. 
We repeat both 
methods on PostgreSQL for TPC-H 10 times, targeting 5\% error and 95\% confidence. In terms of the 
geometric mean speedup, \name outperforms row-level fixed-size sampling by 
93.3$\times$ and underperforms block-level fixed-size sampling by 3.8\%. 
This is because Bernoulli sampling leads to varied sample size which requires 
sampling more data to maintain the same error guarantees, compared to fixed-size 
sampling. However, the performance loss is small since the probability of size 
variation decreases exponentially as the variation amount increases, 
according to the Chernoff Bound on the Binomial distribution \cite{kambo1966exponential}.
\color{black}

\subsection{Latency Decomposition} \label{sec:cost}
We decompose the latency of \name into three parts (1) sample planning 
(\S \ref{subsec:plan}), (2) plan optimization (\S \ref{subsec:optimize}), and 
(3) final execution. We executed each query on PostgreSQL 10 
times and calculated the geometric mean of their latencies. Figure \ref{fig:cost} 
demonstrates the latency proportion of each part. As shown, the sample 
planning via pilot query execution is the major overhead, while the final query
execution constitutes the majority of the total latency.

\color{black}
\subsection{Sensitivity Analysis} \label{sec:sensi}
We conducted a sensitivity analysis of PilotDB's performance across a wide range 
of parameter settings in Procedure \ref{theorem:rel-error}: $\theta_p$, 
$\delta_1$, and $\delta_2$.

\minihead{Impact of the pilot query sampling rate ($\theta_p$)}
We executed TPC-H Query 6 on PostgreSQL with various $\theta_p$ values 
(0.05\%-10\%), aiming for 1\% errors and 95\% confidence. Figure \ref{fig:theta_p} 
shows maximum, minimum, and geometric mean speedups achieved by \name across 10 
executions. We find that the speedup is non-monotonic with respect to 
$\theta_p$: performance declines at low sampling rates due to loose estimations 
and at high rates due to expensive sample planning. Nevertheless, \name achieves 
>6$\times$ speedups consistently.

\minihead{Impact of the failure probability allocation ($\delta_1$, $\delta_2$)}
We execute TPC-H query 6 on PostgreSQL with various $\theta_1$ and $\theta_2$ 
values (0.1\%-4.8\%) , targeting a 1\% error. According to Procedure 
\ref{theorem:rel-error}, we ensure $\delta_1+\delta_2+p'=5\%$ to maintain the 
95\% confidence for the error guarantees. Figure \ref{fig:deltas} shows the 
geometric mean speedup of \name across 10 executions. As shown, \name achieves 
4.8-10.0$\times$ speedups, with the maximum speedup at $\delta_1=0.2\%$ and 
$\delta_2=4.6\%$. Our default setting leads to 21\% lower speedup compared to 
the optimal configuration. For scenarios requiring optimal speedups, we can 
efficiently tune $\delta_1$ and $\delta_2$ with cached pilot query results or
incorporate $\delta_1$ and $\delta_2$ as optimizable parameters during the 
sampling plan optimization.
\color{black}

\section{Related Work}

\minihead{Online AQP} 
Generating samples of large tables upon query arrival is widely studied in prior 
AQP techniques \cite{simple-random-sampling,group-by-sample,sampling-algebra,
analytical-bootstrap,abs,quickr,idea}. Prior work formulated random sampling as 
a standard operation in query processing to estimate aggregates and used 
analytical or bootstrap confidence intervals to measure the estimation error 
\cite{simple-random-sampling,group-by-sample,sampling-algebra,bootstrap,
analytical-bootstrap,abs}. As a step further for complex queries, \textsc{Quickr} 
injects sampling operations in the query plan level and integrate sample 
planning with query optimization to achieve acceleration and a priori error 
guarantees \cite{quickr}. Additionally, \textsc{Idea} reuses previous results to 
accelerate future approximate queries \cite{idea}. More recently, \textsc{Taster} 
combines online and offline methods by caching the online samples for future 
reuse \cite{taster}.

Although existing online AQP systems return estimation errors, they cannot 
provide a priori error guarantees without modifying the underlying DBMS. 
In addition to the DBMS modifications, state-of-the-art methods with a priori 
error guarantees slow down a significant part of queries compared with exact 
execution \cite{quickr,taster} or lead to errors as big as 100\% \cite{quickr}.

\minihead{Offline AQP} 
Prior work developed two types of offline AQP methods: summary-based methods 
\cite{sketch-tutorial,bias-aware-sketch,count-filter-sketch,count-est,
histogram-selectivity,digithist,wavelet,distributed-wavelet} and sampling-based 
method \cite{join-synopses,aqua-sigmod,aqua,icicles,congressional-sample,start,
outlier,sample-selection,blinkdb,sample+seek,verdictdb,baq,taster}.
The primary idea of summary-based offline AQP is to compress or summarize 
columns through numeric transformations. Therefore, they cannot process queries 
with non-numeric columns, such as categorical columns, or with complex 
relational operations, such as join and grouping.

Offline Sampling-based AQP generate offline samples to answer online queries. 
\textsc{Aqua} developed the method of rewriting queries with pre-computed samples 
\cite{join-synopses,congressional-sample,aqua-sigmod,aqua}. Subsequently, 
various optimizations in offline sample creation have been proposed, such as 
weighted sampling \cite{icicles,sample-selection}, stratified sampling 
\cite{start}, and outlier index \cite{outlier}. Prior work has explored 
guaranteeing errors a priori by generating specialized samples for non-nested 
queries \cite{baq}, sparse data distribution \cite{error-bounded-stratified}, 
queries over specific columns \cite{blinkdb}, and queries with specific 
selectivities \cite{sample+seek}. Furthermore, \textsc{VerdictDB} developed 
offline AQP as middleware to avoid modifications to DBMSs \cite{verdictdb}. 

Offline sampling-based AQP methods have two limitations. First, their 
a priori error guarantees are limited to predictable workloads 
\cite{blinkdb,sample+seek,error-bounded-stratified,baq}. For example, \textsc{BlinkDB} 
requires that incoming queries output columns in a pre-defined column set; 
\textsc{Sample+Seek} relies on the pre-knowledge of the query selectivity to 
select the right processing policy (\ie, sample or seek). Moreover, maintaining 
offline samples requires special effort and costs, including regularly 
refreshing samples to ensure statistical correctness and regenerating samples 
when the database changes \cite{aqua,blinkdb,verdictdb}. 

\minihead{Online Aggregation} 
Previous research has explored interactive processing of aggregation queries, 
providing initial results immediately and improving accuracy as more data is 
sampled \cite{online-aggregation,ripple-join,dbo,turbo-dbo,cosmos,gola,deepola,
progressivedb}. OLA, first proposed by Hellerstein et al. 
\cite{online-aggregation}, has been subsequently improved to support join 
queries \cite{ripple-join,wander-join}, scalable processing on large databases 
\cite{dbo,turbo-dbo}, processing multiple queries simultaneously \cite{cosmos}, 
and complex aggregates \cite{gola}. Furthermore, \textsc{ProgressiveDB} explored 
online aggregation as an extension to existing DBMSs using progressive views 
\cite{progressivedb}. Recently, \textsc{DeepOLA} tackled nested queries for 
online aggregation \cite{deepola}. 

Although OLA techniques can continuously update confidence intervals, it is 
invalid to consider the monitored confidence interval as an error guarantee due 
to the problem of peeking at early results \cite{peeking}. Nevertheless, OLA can 
be integrated with the second stage of \name to provide constantly updating 
results, thereby improving the interactivity and user experience.

\minihead{Block Sampling}
In block sampling, data is sampled at the level of physical data blocks or pages, 
a method widely recognized as a more efficient sampling scheme than row-level 
sampling \cite{page-selection,haas2004bi,ci2015efficient,hou1991statistical,
block-sampling-variance,cheng2017bi}. Prior work has studied confidence intervals for aggregates 
computed directly over the output of block sampling \cite{hou1991statistical,
block-sampling-variance,pansare2011online}, block sampling mixed with row-level
sampling (\ie, bi-level sampling) \cite{haas2004bi,cheng2017bi}, and improved the 
statistical efficiency of block sampling with block-level summary statistics 
\cite{page-selection}. However, statistical guarantees for complex approximate 
queries (\eg, nested queries and Join queries) with block sampling have not been 
investigated in literature.

\section{Conclusion}
We propose \name, an online AQP system that achieves (1) a priori 
error guarantees, (2) no maintenance overheads, and (3) no DBMS modifications.
To achieve these properties, we propose a novel online AQP algorithm, \samplingscheme,
based on query rewriting and online sampling. To accelerate queries with
\samplingscheme, we formalize block sampling with new statistical techniques to 
provide guarantees on nested queries and Join queries. Our evaluation shows that 
\name consistently achieves a priori error guarantees and accelerates queries by 
0.92-126$\times$ on various DBMSs.
\section{Acknowledgement}
We are grateful to the CloudLab for providing computing resources for 
experiments \cite{cloudlab}. We thank Andy Luo, Chuxuan Hu, and Lilia Tang for
their feedback and help.

\bibliographystyle{ACM-Reference-Format}
\bibliography{paper}

% \clearpage
\appendix
% Start one-column section for Supplementary Material
\twocolumn[\hsize\textwidth\columnwidth\hsize\csname @twocolumnfalse\endcsname
]
% \section{Detailed Related Work}\label{sec:supp-related-work}
% \input{tex/detailed_related_work.tex}

\section{Additional Experiments}
\subsection{Failures of \samplingscheme with Standard CLT on Queries with Block Sampling} 
\label{subsec:failure}

\begin{figure*}[t!]
    \begin{minipage}{\linewidth}
        \begin{subfigure}[t]{0.25\linewidth}
            \centering
            \includegraphics[width=\linewidth]{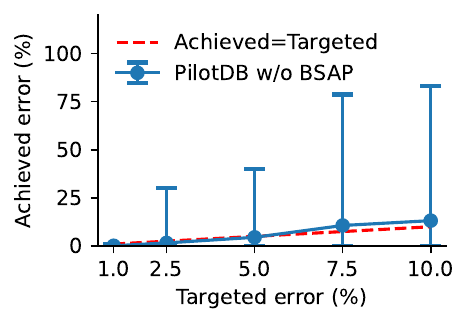}
            \caption{TPC-H.}
            \label{fig:tpch-error-wrong}
        \end{subfigure}%
        \hfill
        \begin{subfigure}[t]{0.25\linewidth}
            \centering
            \includegraphics[width=\linewidth]{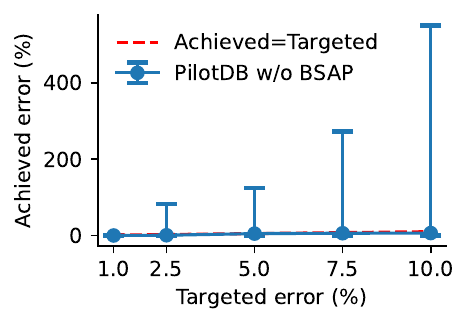}
            \caption{SSB.}
            \label{fig:ssb-error-wrong}
        \end{subfigure}%
        \hfill
        \begin{subfigure}[t]{0.25\linewidth}
            \centering
            \includegraphics[width=\linewidth]{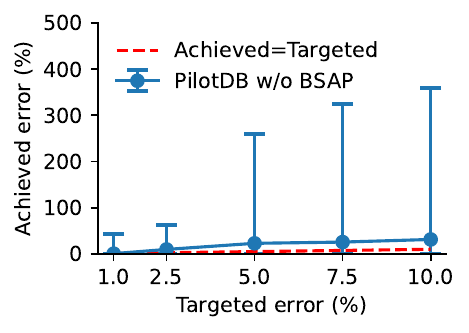}
            \caption{ClickBench.}
            \label{fig:clickbench-error-wrong}
        \end{subfigure}%
        \hfill
        \begin{subfigure}[t]{0.25\linewidth}
            \centering
            \includegraphics[width=\linewidth]{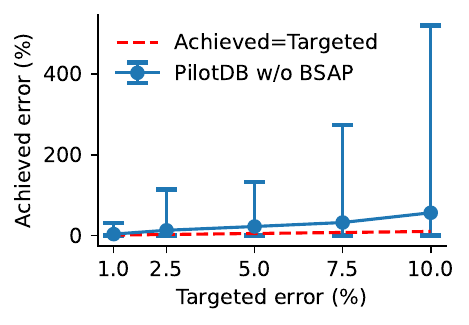}
            \caption{Instacart.}
            \label{fig:instacart-error-wrong}
        \end{subfigure}%
    \caption{\samplingscheme with standard CLT fails to achieves error guarantees 
    on TPC-H, SSB, ClickBench, and Instacart, resulting errors up to 520\% when 
    the targeted error is 10\%. We show the maximum, mean, and minimum errors in 
    20 executions.}
    \label{fig:agg-error-wrong}
    \end{minipage}
\end{figure*}

\begin{figure}
    \centering
    \includegraphics[width=0.55\linewidth]{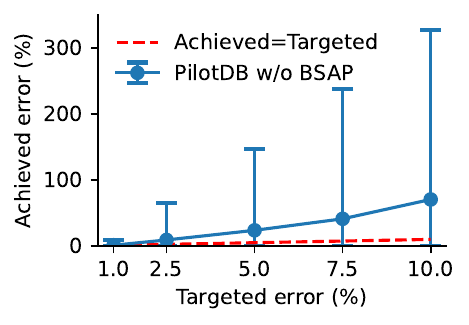}
    \caption{\samplingscheme with standard CLT fails to guarantee errors for
    the DSB benchmark.}
    \label{fig:dsb-fail}
\end{figure}

In this section, we empirically evaluate whether we can apply \samplingscheme
with standard CLT on queries with block sampling. To do that, we treat queries
with block sampling in the same way as treating queries with row-level Bernoulli
sampling. Specifically, $L_\mu$ is calculated using a one-sided confidence interval 
on $\mu$. $U_V[\Theta]$ is calculated using (1) an upper bound of the standard 
deviation, which is governed by a chi-squared distribution, and (2) a lower 
bound of the sample size given $\Theta$, which is governed by a binomial 
distribution \cite{intro-math-stats}. We execute all queries in Table 
\ref{tab:workloads} for 20 times on PostgreSQL. We use the same setting as 
Section \ref{subsec:eval-correct} to calculate and record errors.

We show the errors of all queries in Figures \ref{fig:agg-error-wrong} and 
\ref{fig:dsb-fail}. As shown, the maximum errors of queries with block sampling
can be up to 52$\times$ higher than the user-specified error requirement. 
These experiments demonstrate that standard CLT is not sufficient to analyze 
errors of queries with block sampling in \samplingscheme, motivating the 
development of \bstats.

\subsection{Additional Sensitivity Study} \label{subsec:sensi}
In this section, we analyze the performance of \name under various settings. The 
execution follows the same setting as described in Section \ref{sec:perf}.

\begin{figure}[t]
    \centering
    \begin{minipage}{0.48\linewidth}
        \includegraphics[width=\linewidth]{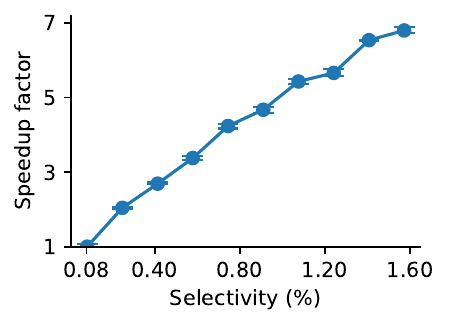}
        \caption{Varying selectivity.}
        \label{fig:selectivity}
    \end{minipage}
    \hfill
    \begin{minipage}{0.48\linewidth}
    \includegraphics[width=\linewidth]{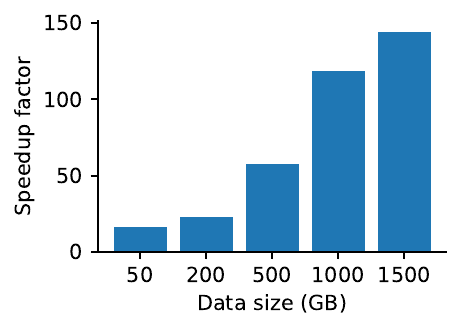}
    \caption{Varying data sizes.}
    \label{fig:data-size}
    \end{minipage}
\end{figure}

\minihead{Impact of Selectivity}
We analyze how the selectivity of the query affects the performance of \name.
Specifically,
    we modify the predicate of TPC-H query 6 to represent different selectivities, targeting a 1\% error.
Figure \ref{fig:selectivity} shows the speedups of \name with various selectivities.
As expected,
    \name achieves better performance with higher selectivity.
This is because we require smaller sampling rates for queries with higher selectivity to achieve the same requested error.

\minihead{Impact of Data Size}
We analyze how the data size affects the performance of \name.
Specifically,
    we execute TPC-H query 6 on PostgreSQL with various scale factors.
Figure \ref{fig:data-size} shows the speedups of \name with 50GB to 1,500GB data.
We observe that speedups of \name increase monotonically regarding the data size,
    demonstrating that \name is suitable for big data analytics.

\section{Proofs of Theoretical Statements} \label{sec:app-proof}
\color{black}
\subsection{Derivation of Probabilistic Bounds} \label{subsec:app-bounds}
We first consider row-level Bernoulli sampling on one table of the query 
in the following lemma.
\begin{lemma}
We assume that the mean estimator over an arbitrary queried table is 
sub-Gaussian. Suppose the pilot query with sampling rate $\theta_p$ results in 
the sample size $n_p$, the sample mean $\hat\mu_p$, and the sample variance 
$\hat\sigma_p^2$. Let $\hat\mu$ be the sample mean obtained via a final query 
with sampling rate $\theta$. When $n_p \rightarrow \infty$, we have 
\begin{align*}
    \mathbb{P}\left[\mu \ge \hat\mu_p - t_{n_p-1,1-\delta_1} \cdot 
    \frac{\hat\sigma_p}{n_p}\right] \ge 1 - \delta_1 \\
    \mathbb{P}\left[Var[\hat\mu] \le \frac{n_p-1}{\chi^2_{n_p-1,1-\delta_2/3}} 
    \frac{\hat\sigma_p^2}{L_N\theta-z_{1-\delta_2/3}\sqrt{L_N\theta(1-\theta)}}\right]
    \ge 1 - \delta_2
\end{align*}
where $t_{n_p-1,1-\delta_1}$ is the $1-\delta_1$ percentile of the Student-t
distribution with $n_p-1$ degrees of freedom, $z_{1-\delta_2/3}$ is the 
$1-\delta_2/3$ percentile of the normal distribution, 
$\chi_{n_p-1,1-\delta_2/3}$ is the $1-\delta_2/3$ percentile of the chi-squared
distribution with $n_p-1$ degrees of freedom, and $L_N$ is the calculated as
\begin{equation*}
    L_N = \left(\sqrt{\frac{n_p}{\theta_p} + z_{1-\delta_2/3}^2\frac{1-\theta_p}{4\theta_p}}
    -\sqrt{z_{1-\delta_2/3}^2\frac{1-\theta_p}{4\theta_p}}\right)^2
\end{equation*}
\end{lemma}

\begin{proof}
The probabilistic lower bound of $\mu$ can be obtained by directly applying 
CLT on the pilot query results \cite{intro-math-stats}. Simultaneously, we can 
obtain a probabilistic upper bound of the variance:
\begin{equation*}
    \mathbb{P}\left[\sigma^2 \le \frac{n_p-1}{\chi^2_{n_p-1,1-\delta_2/3}} 
    \hat\sigma_p^2\right] \ge 1 - \frac{\delta_2}{3}
\end{equation*}
Let $n$ be the sample size of the final query. We can then express the 
variance of $\hat\mu$ as $Var\left[\hat\mu\right] = \frac{\sigma^2}{n}$.
Therefore, we have
\begin{equation}
    \mathbb{P}\left[Var\left[\hat\mu\right] \le \frac{n_p-1}{\chi^2_{n_p-1,1-\delta_2/3}} 
    \frac{\hat\sigma_p^2}{n}\right] \ge 1 - \frac{\delta_2}{3} \label{eq:app-var-bound-1}
\end{equation}

Next, we calculate the lower bound of $n$ given the sampling rate $\theta$. We 
observe that $n$ follows a binomial distribution $Bin(\theta, N)$, where $N$ is 
the total number of units to sample. To simplify calculation, we approximate 
$Bin(\theta, N)$ with a normal distribution 
$\mathcal{N}(N\theta, N\theta(1-\theta))$ 
\cite{verdictdb,intro-math-stats,agresti1998approximate}. Then, we have the 
following probabilistic bound:
\begin{equation}
    \mathbb{P}\left[n \ge N\theta - z_{1-\delta_2/3}\sqrt{N\theta(1-\theta)}\right] 
    \ge 1 - \frac{\delta_2}{3} \label{eq:app-var-bound-2}
\end{equation}

Finally, we calculate the lower bound of $N$ based on the pilot query result. We
also apply the normal approximation to the binomial distribution $Bin(\theta_p, N)$.
We can obtain the following bound:
\begin{equation*}
    \mathbb{P}\left[n_p \le N\theta_p + z_{1-\delta_2/3}\sqrt{N\theta_p(1-\theta_p)}\right] \ge 1 - \frac{\delta_2}{3}
\end{equation*}
Then, we can calculate the following upper bound for $N$:
\begin{equation}
    \mathbb{P}\left[N \ge \left(
        \sqrt{\frac{n_p}{\theta_p}+z_{1-\delta_2/3}^2\frac{1-\theta_p}{4\theta_p}}
        -\sqrt{z_{1-\delta_2/3}^2\frac{1-\theta_p}{4\theta_p}}\right)^2 \right] 
        \ge 1 - \frac{\delta_2}{3} \label{eq:app-var-bound-3}
\end{equation}

We can chain Inequalities \ref{eq:app-var-bound-1}, \ref{eq:app-var-bound-2}, 
and \ref{eq:app-var-bound-3} and obtain the final probabilistic upper bound
for $Var[\hat\mu]$.
\end{proof}
\minihead{Discussion of $n_p \rightarrow \infty$} Our derivation of probabilistic
bounds is based on CLT, which holds asymptotically as $n_p \rightarrow \infty$. 
In consistent with literature 
\cite{abae,inquest,quickr,hou1991error,hertzog2008considerations}, we consider
$n_p=30$ is sufficiently large. To have more conservative guarantees, we can 
use the Chernoff bounds for sub-Gaussian random variables to derive the bounds 
\cite{sub-gaussian,buldygin1980sub}.
\color{black}

\subsection{Proof of Theorem \ref{theorem:rel-error}}\label{proof:rel-error}
\begin{proof}
We first recap the definitions of $L_\mu$, $U_V[\Theta]$, and the CLT-based 
confidence interval on $\mu$:
\begin{align*}
    \mathbb{P}\left[\mu \ge L_\mu\right] &\ge 1 - \delta_1 \\
    \mathbb{P}\left[Var[\hat\mu](\Theta) \le U_V[\Theta]\right] &\ge 1 - \delta_2 \\
    \mathbb{P}\left[
        \hat\mu - z_{(1+p')/2}\sqrt{Var[\hat\mu]} 
        \le \mu \le 
        \hat\mu + z_{(1+p')/2}\sqrt{Var[\hat\mu]}
    \right] &\ge p'
\end{align*}
Then, based on Boole's Inequality, we have
\begin{align*}
\mathbb{P} & \big[
    \left(\mu \ge L_\mu\right) \vee 
    \left(Var[\hat\mu](\Theta) \le U_V[\Theta]\right) \\
    & \vee \left(\left|\hat\mu - \mu\right| \le z_{(1+p')/2}\sqrt{Var[\hat\mu](\Theta)}\right) \big] \\
& \ge 1 - \delta_1 - \delta_2 - (1-p') = p' - \delta_1 - \delta_2
\end{align*}
which implies that
\begin{equation*}
\mathbb{P} \left[
    \left|\frac{\hat\mu - \mu}{\mu} \le \frac{z_{(1+p')/2}\sqrt{U_V[\Theta]}}{L_\mu}\right|
\right] \ge p' - \delta_1 - \delta_2
\end{equation*}
Therefore, if we let
\begin{equation*}
    \frac{z_{(1+p')/2}\sqrt{U_V[\Theta]}}{L_\mu} \le e \quad
    \text{and} \quad
    p' - \delta_1 - \delta_2 = p
\end{equation*}
user's error specification will be satisfied, which is
\begin{equation}
    \mathbb{P} \left[
        \left|\frac{\hat\mu - \mu}{\mu} \le e\right|
    \right] \ge p
\end{equation}
\end{proof}

\subsection{Proof of the Rules in Table \ref{tab:error-prop}}\label{proof:error-prop}
\begin{lemma}(\textsc{Multiplication})
Let $\mu$ be a quantity calculated as the product of two positive quantities $\mu_1$ 
and $\mu_2$ (\ie, $\mu=\mu_1\mu_2$). We estimate $\mu_1$ with $\hat \mu_1$ and $\mu_2$ 
with $\hat \mu_2$. Let $e_{\mu_1} (<1)$ be the relative error between $\hat \mu_1$ 
and $\mu_1$, and $e_{\mu_2} (<1)$ be the relative error between $\hat \mu_2$ and 
$\mu_2$. The relative error between $\mu$ and $\hat \mu_1 \hat \mu_2$ has an upper 
bound of $e_{\mu_1} + e_{\mu_2} + e_{\mu_1} \cdot e_{\mu_2}$.
\end{lemma}
\begin{proof}
By definition of the relative error and the positiveness of $\mu_1$, we have
\begin{equation}
\left| \frac{\hat\mu_1 - \mu_1}{\mu_1} \right| \le e_{\mu_1} \quad 
\Leftrightarrow \quad (1-e_{\mu_1}) \mu_1 \le \hat \mu_1 \le (1+e_{\mu_1}) \mu_1 \notag
\end{equation}
Similarly, we have
\begin{equation}
\left| \frac{\hat\mu_2 - \mu_2}{\mu_2} \right| \le e_{\mu_2} \quad 
\Leftrightarrow \quad (1-e_{\mu_2}) \mu_2 \le \hat \mu_2 \le (1+e_{\mu_2}) \mu_2 \notag
\end{equation}
Then, we have
\begin{gather}
(1-e_{\mu_1})(1-e_{\mu_2})\mu_1\mu_2 \le \hat \mu_1 \hat \mu_2 \le (1+e_{\mu_1})(1+e_{\mu_2})\mu_1\mu_2 \notag \\
\Leftrightarrow  (e_{\mu_1}e_{\mu_2} - e_{\mu_1} - e_{\mu_2})\mu_1\mu_2 \le \hat \mu_1 \hat \mu_2 - \mu_1\mu_2 \\
\le (e_{\mu_1}e_{\mu_2} + e_{\mu_1} + e_{\mu_2})\mu_1\mu_2 \notag \\
\Leftrightarrow \left| \frac{\hat \mu_1 \hat \mu_2 - \mu_1\mu_2}{\mu_1\mu_2} \right| \le e_{\mu_1} + e_{\mu_2} + e_{\mu_1} \cdot e_{\mu_2} \notag
\end{gather}
\end{proof}

\begin{lemma}(\textsc{Division})
Let $\mu$ be a quantity calculated as the ratio of two positive quantities $\mu_1$ and 
$\mu_2$ (\ie, $\mu=\mu_1/\mu_2$). We estimate the $\mu_1$ with $\hat\mu_1$ and $\mu_2$ with $\hat\mu_2$. 
Let $e_{\mu_1} (<1)$ be the relative error between $\hat\mu_1$ and $\mu_1$, and $e_{\mu_2} (<1)$ 
be the relative error between $\hat\mu_2$ and $\mu_2$. The relative error between $\mu$ 
and $\hat\mu_1/\hat\mu_2$ has an upper bound of $\frac{e_{\mu_1}+e_{\mu_2}}{1+\min(e_{\mu_1}, e_{\mu_2})}$.
\end{lemma}
\begin{proof}
    By the definition of the relative error and the positiveness of $\mu_1$ and $\mu_2$, we have
    \begin{gather}
\frac{1-e_{\mu_1}}{1+e_{\mu_2}}\frac{\mu_1}{\mu_2} \le \frac{\hat\mu_1}{\hat\mu_2} \le \frac{1+e_{\mu_1}}{1-e_{\mu_2}}\frac{\mu_1}{\mu_2} \notag \\
\Leftrightarrow -\frac{e_{\mu_1} + e_{\mu_2}}{1+e_{\mu_1}}\frac{\mu_1}{\mu_2}\le 
    \frac{\hat\mu_1}{\hat\mu_2} - \frac{\mu_1}{\mu_2} \le \frac{e_{\mu_1} + e_{\mu_2}}{1+e_{\mu_2}}\frac{\mu_1}{\mu_2} \notag \\
\Rightarrow \left| \frac{\hat\mu_1/\hat\mu_2 - \mu_1/\mu_2}{\mu_1/\mu_2} \right| \le \frac{e_{\mu_1}+e_{\mu_2}}{1+\min(e_{\mu_1}, e_{\mu_2})} \notag
    \end{gather}
\end{proof}

\begin{lemma}(\textsc{Addition})
Let $\mu$ be a quantity calculated as the linear combination of two positive 
quantities $\mu_1$ and $\mu_2$. Namely, $\mu = \lambda_1 \mu_1 + \lambda_2 \mu_2$, 
where $\lambda_1$ and $\lambda_2$ are positive. We estimate $\mu_1$ with $\hat\mu_1$ 
and $\mu_2$ with $\hat\mu_2$. Let $e_{\mu_1} (<1)$ be the relative error between 
$\hat\mu_1$ and $\mu_1$, and $e_{\mu_2} (<1)$ be the relative error between 
$\hat\mu_2$ and $\mu_2$. The relative error between $\mu$ and 
$\lambda_1 \hat\mu_1+\lambda_2\hat\mu_2$ has an upper bound of $\max(e_{\mu_1}, e_{\mu_2})$.
\end{lemma}
\begin{proof}
By the definition of the relative error and the positiveness of $\mu_1$ and $\mu_2$, we have
\begin{gather}
(1-e_{\mu_1})\lambda_1\mu_1 + (1-e_{\mu_2})\lambda_2\mu_2 \le 
    \lambda_1 \hat\mu_1 + \lambda_2\hat\mu_2 \\
    \le (1+e_{\mu_1})\lambda_1\mu_1 + (1+e_{\mu_2})\lambda_2\mu_2 \notag \\
\Leftrightarrow -\lambda_1e_{\mu_1} \mu_1 - \lambda_2e_{\mu_2} \mu_2 \le 
    \lambda_1\hat\mu_1 + \lambda_2\hat\mu_2 - (\lambda_1\mu_1 + \lambda_2\mu_2) \\
    \le \lambda_1e_{\mu_1} \mu_1+\lambda_2e_{\mu_2} \mu_2 \notag \\
\Rightarrow \left|
    \frac{ \lambda_1\hat\mu_1 + \lambda_2\hat\mu_2 - (\lambda_1\mu_1 + \lambda_2\mu_2)}{\lambda_1\mu_1 + \lambda_2\mu_2}
 \right| \le \max(e_{\mu_1}, e_{\mu_2}) \notag
    \end{gather}
\end{proof}

\subsection{Proof of Lemma \ref{lemma:group-error}}\label{proof:group-error}
\begin{lemma}
    For a table $T$ with block size $b$, block sampling with a sampling rate $\theta$ 
    satisfying the condition below ensures that the probability of missing a group of 
    size greater than $g$ is less than $p_f$.
    \begin{equation}
        \theta \ge 
        1 - \left(1-\left(1-p_f\right)^{\lceil g/b \rceil/|T|}\right)^{1/\lceil g/b \rceil} 
        \label{eq:group-error}
    \end{equation}
    \end{lemma}
\begin{proof}
Based on the meta-information, we calculate the number of blocks in $T$ as $N_b=|T|/b$. 
Because each group has at least $g$ rows, each group takes at least 
$n_b=\left\lceil g/b\right\rceil$ blocks. Suppose there are $t$ groups with size 
larger than $g$. Let $n_i$ be the number of blocks taken by the $i$-th group. 
We then have the following constraints
\begin{gather}
        t \le \frac{|T|}{n_0} \label{eq:const-group-num} \\
        \forall_{1\le i\le t},\ n_i \ge n_0  \label{eq:const-group-size}
\end{gather}
Based on the process of block sampling, we can calculate the probability of 
including $i$-th group as following
\begin{equation}
\mathbb{P}\left[\text{include group i}\right] = 1- \mathbb{P}\left[\text{miss group i}\right] = 1-(1-\theta)^{n_i} \notag
\end{equation}
Next, we calculate the probability of including groups $i$ and $j$ ($i \neq j$). 
Suppose group $i$ and group $j$ share $k$ blocks ($k \ge 0$), then the probability of including two groups has the following lower bound.
\begin{align}
    & \mathbb{P}\left[(\text{include group i}) \wedge  (\text{include group j})\right] \notag\\
    &= \mathbb{P}\left[\text{include group i}\ |\ \text{include group j}\right] \cdot \mathbb{P}\left[\text{include group j}\right] \notag \\
    &= \bigl(\mathbb{P}\left[\text{include group i}\right] \cdot \mathbb{P}\left[\text{not include shared blocks}\ |\ \text{include group j}\right] \notag\\
    & \quad + \mathbb{P}\left[\text{include shared blocks}\ |\ \text{include group j}\right] \bigr) \cdot \mathbb{P}\left[\text{include group j}\right] \notag \\
    &= \left(\left(1-(1-\theta)^{n_i}\right) \cdot \frac{n_j-k}{n_j} + \frac{k}{n_j}\right) \cdot \left(1-(1-\theta)^{n_j}\right) \notag \\
    &\ge \left(1-(1-\theta)^{n_i}\right) \cdot \left(1-(1-\theta)^{n_j}\right) \notag
\end{align}
We extend the results to all groups and calculate lower bound of the probability 
of including all groups
\begin{equation}
    \mathbb{P}\left[\text{include all groups}\right] \ge \prod_{i=1}^t \left(1-(1-\theta)^{n_i}\right) \notag
\end{equation}
Applying the constraints in \ref{eq:const-group-num} and \ref{eq:const-group-size}, we can have the following lower bound for the probability of including all groups
\begin{align*}
    \mathbb{P}\left[\text{include all groups}\right] 
    & \ge \prod_{i=1}^t \left(1-(1-\theta)^{n_0}\right) \\
    & \ge \prod_{i=1}^{|R|/n_0} \left(1-(1-\theta)^{n_0}\right) \\
    & = \left(1-(1-\theta)^{\lceil g/b \rceil}\right)^{|T|/\lceil g/b \rceil} \notag
\end{align*}
Therefore, if the sampling rate $\theta$ satisfies
\begin{equation}
    \theta \ge 1 - \left(1-\left(1-p\right)^{\lceil g/b \rceil/|T|}\right)^{1/\lceil g/b \rceil} \notag
\end{equation}
then
\begin{equation}
    \mathbb{P}\left[\text{miss a group}\right] = 1-\mathbb{P}\left[\text{include all groups}\right] \le p \notag
\end{equation}
\end{proof}

\subsection{Proof of Lemma \ref{lemma:comp-stats}}\label{proof:stats-eff}
\begin{proof}
Let $n_b$ be the number of blocks in the block sample and $n_r$ be the number of
rows in the row-level uniform sample. Then, the estimated mean using the block 
sample is
\begin{equation*}
\hat\mu_b = \frac{1}{b\cdot n_b}\sum_{i=1}^{n_b\cdot b} X_{i}
\end{equation*}
The estimated mean using the row-level uniform smaple is
\begin{equation*}
\hat\mu_u = \frac{1}{n_r}\sum_{i=1}^{n_r} X_{i}
\end{equation*}
We find that both estimators are unbiased:
\begin{align*}
&\mathbb{E}\left[\hat\mu_b\right] 
    = \frac{1}{n_b} \sum_{i=1}^{n_b} \mathbb{E}\left[\frac{1}{b}\sum_{i=j\cdot b}^{(j+1)\cdot b - 1} X_i\right]
    = \frac{1}{n_b} \sum_{i=1}^{n_b} \frac{1}{Nb} \sum_{i=1}^{N\cdot b} X_i
    = \mu \\
&\mathbb{E}\left[\hat\mu_r\right] = \frac{1}{n_r}\sum_{i=1}^{n_r}\mathbb{E}\left[ X_{i}\right] = \mu
\end{align*}
Therefore, the variance of the estimator is essentially the expected squared error.
To achieve the same accuracy in expectation, we let the variance of two estimators 
be the same. That is,
\begin{equation}
    Var[\hat\mu_b] = Var[\hat\mu_r] = \frac{1}{n_r} Var[X_i]
\end{equation}
Moreover, based on the law of total variance, we can decompose the variance of $\hat\mu_b$
in the following way:
\begin{equation*}
Var[\hat\mu_b] = \frac{1}{n_b} Var\left[\frac{1}{b}\sum_{i=j\cdot b}^{(j+1)\cdot b - 1} X_i\right]
= \frac{1}{n_b} \left( Var\left[X_i\right] - \mathbb{E}\left[\sigma_j^2\right] \right)
\end{equation*}
Finally, we can derive the ratio between the sample size of block sampling and 
that of row-level uniform sampling to achieve the same accuracy in expectation as
follows:
\begin{align}
\frac{b\cdot n_b}{n_r} = b\left(1 - \frac{\mathbb{E}\left[\sigma_j^2\right]}{Var\left[X_i\right]}\right)
\end{align}
\end{proof}

\subsection{Proof of Proposition \ref{prop:identical-agg}}\label{proof:identical-agg}
\begin{proof}
Let $\tilde{R}=\{R_1^*, \ldots, R_k^*\}$ be the set of tables on which the aggregate 
function $f$ produces $x$. Then, we have
\begin{align*}
\mathbb{P}\left[f\biggl(\mathcal{S}_1(\{T_1, \ldots, T_n\})\biggr) = x\right]
    &= \mathbb{P}\left[\bigcup_{i=1}^k \mathcal{S}_1(\{T_1, \ldots, T_n\}) = R_i^*\right] \\
    &= \sum_{i=1}^k \mathbb{P}\left[\mathcal{S}_1(\{T_1, \ldots, T_n\}) = R_i^*\right] \\
\end{align*}
Similarly, we have
\begin{equation*}
\mathbb{P}\left[f\left(\mathcal{S}_2(\{T_1, \ldots, T_n\})\right) = x\right] 
    = \sum_{i=1}^k \mathbb{P}\left[\mathcal{S}_2(\{T_1, \ldots, T_n\}) = R_i^*\right]
\end{equation*}
Since $\mathcal{S}_1$ and $\mathcal{S}_2$ are equivalent, we then have
\begin{equation*}
    \mathbb{P}\left[f\left(\mathcal{S}_1(\{T_1, \ldots, T_n\})\right) = x\right] 
    = \mathbb{P}\left[f\left(\mathcal{S}_2(\{T_1, \ldots, T_n\})\right) = x\right] 
\end{equation*}
\end{proof}

\subsection{Proof of Propositions \ref{theorem:select-comm}, \ref{theorem:join-comm}, and \ref{theorem:union-comm}} \label{proof:equiv}

\minihead{Proposition \ref{theorem:select-comm} (Selection)}
\begin{proof}
First, we describe two events whose probability is zero and independent of the order of block sampling and selection. Suppose the input table $T$ has $N$ blocks: $P_1, \ldots, P_N$. Let $T^*$ be the result table of applying block sampling and selection on $T$.
\begin{align}
    E_1 &:= \ \exists_{r \in T^*},\ \phi(r)=0 \notag \\
    E_2 &:= \ \forall_{i \in [1,N]}\ \exists_{r_1, r_2\in P_i},\ \phi(r_1)=\phi(r_2)=1\ \texttt{AND}\ r_1 \in T^*\ \texttt{AND}\ r_2 \notin T^* \notag
\end{align}
As long as the selection operation is applied, every row in the result table must evaluate the predicate $\phi$ to 1. Therefore, the probability of $E_1$ is $0$. Furthermore, block sampling ensures that if one row is sampled, the whole block is sampled. Therefore, it is not likely to have two rows from the same block satisfying the predicate, but only one of them is in the result table. Namely, the probability of $E_2$ is $0$. 

Next, we analyze the probability distributions excluding events $E_1$ and $E_2$. Let $T'=\{P'_1, \ldots, P'_{N'}\}$ be the subset of $T$ satisfying $\phi$, where each block $P'_i$ contains non-zero rows satisfying $\phi$. If we exclude events $E_1$ and $E_2$, the result table must be a subset of blocks in $T'$, regardless of the order of operations. We can calculate the inclusion probability of $P'_i$ in $T^*$ for different orders of operations. It turns out that the inclusion probability is independent of the operation order.
\begin{align}
    & \mathbb{P}\left[P'_i \in \sigma_\psi(\mathcal{B}_\theta(T))\right] = \mathbb{P}\left[P_i \in \mathcal{B}_\theta(T)\right] = \theta \notag \\
    & \mathbb{P}\left[P'_i \in \mathcal{B}_\theta(\sigma_\psi(T))\right] = \mathbb{P}\left[P'_i \in \mathcal{B}_\theta(T')\right] = \theta \notag
\end{align}

Since blocks are independent of each other in the process of selection and block sampling, for the result table $T^*=\{P'_1, \ldots, P'_n\}$, we have
\begin{align}
    & \mathbb{P}\left[\sigma_\psi(\mathcal{B}_\theta(T)) = T^*\right] \notag \\
    & = \mathbb{P}\left[\left(\bigcap_{i=1}^n P'_i \in \sigma_\psi(\mathcal{B}_\theta(T))\right) \bigcap\left(\bigcap_{i=n+1}^{N'} P'_i \notin \sigma_\psi(\mathcal{B}_\theta(T))\right)\right] \notag \\
    & = \theta^n (1-\theta)^{N'-n}; \notag \\
    & \mathbb{P}\left[\mathcal{B}_\theta(\sigma_\psi(T)) = T^*\right] \notag \\
    & = \mathbb{P}\left[\left(\bigcap_{i=1}^n P'_i \in \mathcal{B}_\theta(\sigma_\psi(T))\right) \bigcap\left(\bigcap_{i=n+1}^{N'} P'_i \notin \mathcal{B}_\theta(\sigma_\psi(T))\right)\right] \notag \\
    & = \theta^n (1-\theta)^{N'-n} \notag
\end{align}
Therefore, for any result table, the probability is independent of the order of selection and block sampling. Namely, selection and block sampling commutes.
\end{proof}

\minihead{Proposition \ref{theorem:join-comm} (Join)}
\begin{proof}
Since block sampling commutes with selection, it is sufficient to prove that block sampling also commutes with cross-product. Without loss of generality, we assume block sampling is executed on $T_1$. We first describe an event whose probability is zero and is independent of the order of block sampling and cross-product. Suppose $T_1$ has $N$ blocks: $T_1 = \{P_1. \ldots, P_N\}$. Let $T^*$ be the result table after applying block sampling and cross-product over $T_1$ and $T_2$.
\begin{equation}
    E := \quad \forall_{i \in [1, N]} \ \exists_{r_1, r_2 \in P_i},\ (r_1 \Join T_2) \in T^* \ \texttt{AND} \ (r_2 \Join T_2) \notin T^* \notag
\end{equation}
Block sampling ensures that if one row is sampled, the whole block is sampled. Therefore, it is not likely to have one row in the result table while other rows from the same block are not. Namely, the probability of $E$ is $0$.

Next, we analyze the probability distributions excluding the event $E$. Let $T' = \{P'_1, \ldots, P'_N\}$ be the full cross-product where $P'_i=P_i\Join T_2$ (\ie, the cross-product of block $P_i$ and table $T_2$). If we exclude event $E$, the result table $T$ must be a subset of blocks in $T'$. We can calculate the inclusion probability of $P'_i$ in $T$. It turns out that the inclusion probability is independent of the operation order.
\begin{align}
    & \mathbb{P}\left[P'_i \in \left(\mathcal{B}_\theta(T_1) \Join T_2\right)\right] = \mathbb{P}\left[P_i \in \mathcal{B}_\theta(T_1)\right] = \theta \notag \\
    & \mathbb{P}\left[P'_i \in \mathcal{B}_\theta(T_1 \Join T_2)\right] = \mathbb{P}\left[P'_i \in \mathcal{B}_\theta(T')\right] = \theta \notag
\end{align}

Since blocks are independent of each other in the process of cross-product and block sampling, for the arbitrary result table $T^*=\{P'_1, \ldots, P'_n\}$, we have
\begin{align}
    & \mathbb{P}\left[\left(\mathcal{B}_\theta(T_1) \Join T_2\right) = T^*\right] \notag \\
    & = \mathbb{P}\left[\left(\bigcap_{i=1}^n P'_i \in \left(\mathcal{B}_\theta(T_1) \Join T_2\right)\right) \bigcap\left(\bigcap_{i=n+1}^{N} P'_i \notin\left(\mathcal{B}_\theta(T_1) \Join T_2\right)\right)\right] \notag \\
    & = \theta^n (1-\theta)^{N-n}; \notag \\
    & \mathbb{P}\left[\mathcal{B}_\theta(T_1 \Join T_2) = T^*\right] \notag \\ 
    & = \mathbb{P}\left[\left(\bigcap_{i=1}^n P'_i \in \mathcal{B}_\theta(T_1 \Join T_2)\right) \bigcap\left(\bigcap_{i=n+1}^{N} P'_i \notin \mathcal{B}_\theta(T_1 \Join T_2)\right)\right] \notag \\
    & = \theta^n (1-\theta)^{N-n} \notag
\end{align}

Therefore, for any result table, the probability is independent of the order of cross-product and block sampling. Namely, cross-product and block sampling commute. Since selection and block sampling commute, join and block sampling commute.
\end{proof}

\minihead{Proposition \ref{theorem:union-comm} (Union)}
\begin{proof}
First, we describe an event whose probability is zero and independent of the order of union and block sampling. Suppose $T_1$ has $N_1$ blocks: $T_1 = \{P^{(1)}_{1}, \ldots, P^{(1)}_{N_1}\}$, and $T_2$ has $N_2$ blocks: $T_2 = \{P^{(2)}_1, \ldots, P^{(2)}_{N_2}\}$. Let $T^*$ be the result table after applying block sampling and union over $T_1$ and $T_2$.
\begin{align}
    & E_1 := \quad \forall_{i \in [1, N_1]} \ \exists_{r_1, r_2 \in P^{(1)}_i},\ r_1 \in T^* \ \texttt{AND} \ r_2 \notin T^* \notag \\
    & E_2 := \quad \forall_{i \in [1, N_2]} \ \exists_{r_1, r_2 \in P^{(2)}_i},\ r_1 \in T^* \ \texttt{AND} \ r_2 \notin T^* \notag
\end{align}
Block sampling ensures that if one row is sampled, the whole block is sampled. Therefore, it is not likely to have
one row in the result table while other rows from the same block are not. Namely, both the probability of $E_1$ and the probability of $E_2$ are zero.

Next, we analyze the probability distribution excluding $E_1$ and $E_2$. Let $T' = \{P^{(1)}_{1}, \ldots, P^{(1)}_{N_1}, P^{(2)}_1, \ldots, P^{(2)}_{N_2}\}$ be the union of $T_1$ and $T_2$. If we exclude $E_1$ and $E_2$, the result table $T$ must be a subset of blocks in $T'$. We can calculate the inclusion probability of $P^{(1)}_i$ and $P^{(2)}_i$ in $T$. It turns out that the inclusion probability is independent of the operation order. Without loss of generality, we only show the calculation of the inclusion probability of $P^{(1)}_i$.
\begin{align}
    & \mathbb{P}\left[P^{(1)}_i \in \left(\mathcal{B}_\theta(T_1) \cup \mathcal{B}_\theta(T_2)\right)\right] = \mathbb{P}\left[P^{(1)}_i \in \mathcal{B}_\theta(T_1)\right] = \theta \notag \\
    & \mathbb{P}\left[P^{(1)}_i \in \mathcal{B}_\theta(T_1 \cup T_2)\right] = \theta \notag
\end{align}

Since blocks are independent of each other in the process of union and block sampling, for the arbitrary result table:
$$T^* = \{P^{(1)}_1, \ldots, P^{(1)}_{n_1}, P^{(2)}_1, \ldots, P^{(2)}_{n_2}\}$$
we have
\begin{align}
    &\mathbb{P}\left[\left(\mathcal{B}_\theta(T_1) \cup \mathcal{B}_\theta(T_2)\right) = T^*\right] \notag \\
    &=\mathbb{P}\biggl[\left(\bigcup_{i=1}^{n_1} P_i^{(1)} \in \mathcal{B}_\theta(T_1)\right) \bigcup \left(\bigcup_{i=n_1+1}^{N_1} P_i^{(1)} \notin \mathcal{B}_\theta(T_1)\right) \notag \\
    & \quad \quad \bigcup \left(\bigcup_{i=1}^{n_2} P_i^{(2)} \in \mathcal{B}_\theta(T_2)\right) \bigcup \left(\bigcup_{i=n_2+1}^{N_2} P_i^{(2)} \notin \mathcal{B}_\theta(T_2)\right)\biggr] \notag \\
    &= \theta^{n_1} \cdot (1-\theta)^{N_1-n_1} \cdot \theta^{n_2} \cdot (1-\theta)^{N_2-n_2} \notag \\
    & = \theta^{n_1+n_2} \cdot (1-\theta)^{N_1 + N_2 - n_1 - n_2}; \notag \\
    &\mathbb{P}\left[\mathcal{B}_\theta(T_1 \cup T_2) = T^*\right] \notag  \\
    & = \mathbb{P}\biggl[\left(\bigcup_{k=1}^2 \bigcup_{i=1}^{n_k} P_i^{(k)} \in \mathcal{B}_\theta(T_1 \cup T_2)\right) \notag \\
    & \quad \quad \bigcup \left(\bigcup_{k=1}^2 \bigcup_{i=n_k+1}^{N_k} P_i^{(k)} \notin \mathcal{B}_\theta(T_1 \cup T_2)\right)\biggr] \notag \\
    &= \theta^{n_1 + n_2} (1-\theta)^{N_1 - n_1 + N_2 - n_2} \notag
\end{align}
Therefore, for any result table, the probability is independent of the order of union and block sampling. Namely, union and block sampling commute.
\end{proof}

\subsection{Proof of Theorem \ref{theorem:join-clt}}\label{proof:join-clt}
\begin{proof}
We consider the Join of $k$ block samples as $n$ independent and orthogonal 
blcok sampling procedures on the full cross product. Then, the final convergence 
can be proved by induction. We set up the notations as follows.

Let $w_{t,j}(i_1, \ldots, i_t)$ be the sum of Join with fixed indices,
$(i_1, \ldots, i_t)$ on the first $t$ table, a fixed index $j$ on the $t+1$-th 
table, and all of the last $k-1-t$ tables. Namely,
\begin{equation}
    w_{t,j}(i_1, \ldots, i_t) = \sum_{i_{t+2}=1}^{N_{t+2}} \cdots \sum_{i_{k}}^{N_k} \mathcal{J}(i_1, \ldots, i_t, j, i_{t+2}, \ldots, i_k)
\end{equation}
Then, let $w_t$ be the vector of of $w_{t,j}$s with all possible $j$s:
\begin{equation}
    \vec w_t(i_1, \ldots, i_t) = (w_{t,1}(i_1, \ldots, i_t), \ldots, w_{t, N_{t+1}}(i_1, \ldots, i_t))^\top
\end{equation}
Suppose we know the sampling results of the first $t$ tables, we denote $\hat W_t$
as the mean of the result:
\begin{equation*}
    \hat W_{n_1, \ldots, n_{t}}^{(t)} = \left(\prod_{i=1}^t \frac{1}{n_i}\right) \sum_{i_1=1}^{n_1} \cdots \sum_{i_t=1}^{n_t} \vec w_t(i_1, \ldots, i_t)
\end{equation*}
Let $u_{t,i}$ be a one-hot vector of length $N_{t+1}$, where the $i$-th element 
is 1 and the rest of elements are 0:
\begin{equation}
    \vec u_{t,i} = (\overbrace{\underbrace{0, \ldots, 0}_{i-1}, 1, 0, \ldots, 0}^{N_{t+1}})^\top
\end{equation}
Let $\tilde{u}_t$ be the set of $N_{t+1}$ $u_{t,i}$s and $U_{n_{t+1}}^{(t)}$ be 
the mean of $n_{t+1}$ vectors randomly sampled from $\tilde{u}_t$ without replacement.
Namely,
\begin{equation*}
    \tilde{u}_t = \{u_{t,1}, \ldots, u_{t,N_{t+1}}\}, \quad 
    \hat U_{n_{t+1}}^{(t)} = \frac{1}{n_{t+1}} \sum_{i=1}^{n_{t+1}} \vec u_{t,i}
\end{equation*}
We denote $\hat Z_{n_1, \ldots, n_{t+1}}$ as the inner product of $\hat W_{n_1, \ldots, n_{t}}^{(t)}$
and $\hat U_{n_{t+1}}^{(t)}$:
\begin{equation*}
    \hat Z_{n_1, \ldots, n_{t+1}} = \left(\hat W_{n_1, \ldots, n_{t}}^{(t)}\right)^\top \hat U_{n_{t+1}}^{(t)}
\end{equation*}
Then, we find that the mean of the Join result over block samples $\hat\mu$ can 
be represented as $\hat Z_{n_1, \ldots, n_t}$ with $t=k$. Next, we prove the convergence of
$\hat Z_{n_1, \ldots, n_t}$ by induction on $t$.

\minihead{Base case} When $t=0$, we have
\begin{gather}
    w_{0,j} = \sum_{i_2=1}^{N_2} \cdots \sum_{i_{k}}^{N_k} \mathcal{J}(j, i_2, \ldots, i_k) \\
    \hat Z_{n_1} = \frac{1}{n_1} \sum_{i_1}^{n_1} \sum_{i_2=1}^{N_2} \cdots\sum_{i_{k}}^{N_k} \mathcal{J}(i_1, i_2, \ldots, i_k)
\end{gather}
In this case, the convergence of $\hat Z_{n_1}$ follows the standard CLT. Namely,
with $\mu^{(1)} = \frac{1}{N_1} \sum_{i_1}^{N_1} \cdots\sum_{i_{k}}^{N_k} \mathcal{J}(i_1, \ldots, i_k)$,
we have
\begin{equation*}
    \hat Z_{n_1} - \mu^{(1)} \xrightarrow{D} \mathcal{N}(0, Var\left[\hat Z_{n_1}\right])
\end{equation*}

\minihead{Induction step}
Suppose the convergence of $\hat Z_{1n_1, \ldots, n_t}$ holds, \ie,
\begin{equation*}
    \hat Z_{n_1, \ldots, n_t} - \mu^{(t)} \xrightarrow{D} \mathcal{N}\left(0, Var\left[\hat Z_{n_1, \ldots, n_t}\right]\right)
\end{equation*}
Then, for the case of $t+1$, we have the following according to multivariate CLT.
\begin{equation*}
    \hat W_{n_1, \ldots, n_t}^{(t)} - W_{n_1, \ldots, n_t}^{(t)} \xrightarrow{D} 
    \mathcal{N}(0, \Sigma_W)
\end{equation*}
where $\Sigma_W$ is the covariance matrix of $\hat W_{n_1, \ldots, n_t}^{(t)}$ and
\begin{equation*}
    W_{n_1, \ldots, n_t}^{(t)} = \left(\prod_{i=1}^t \frac{1}{N_i}\right) \sum_{i_1=1}^{N_1} \cdots \sum_{i_t=1}^{N_t} \vec w_t(i_1, \ldots, i_t)
\end{equation*}
According to the standard multivariate CLT, we have
\begin{equation}
    \hat U_{n_{t+1}}^{(t)} - U_{n_{t+1}}^{(t)} \xrightarrow{D} \mathcal{N}\left(0, \Sigma_U\right)
\end{equation}
where $\Sigma_U$ is the covariance matrix of $\hat U_{n_{t+1}}^{(t)}$ and
\begin{equation*}
    U_{n_{t+1}}^{(t)} = \underbrace{\left(\frac{1}{N_{t+1}}, \ldots, \frac{1}{N_{t+1}}\right)}_{N_{t+1}}
\end{equation*}
Since the block sampling procedure on table $t+1$ is indenpendent of the block
sampling procedures on previous $t$ tables, we can apply the Cramer-Wold Theorem,
which leads to
\begin{equation*}
    \left(\hat W_{n_1, \ldots, n_{t}}^{(t)}, \hat U_{n_{t+1}}^{(t)}\right) - 
    \left(W_{n_1, \ldots, n_t}^{(t)}, U_{n_{t+1}}^{(t)}\right) \xrightarrow{D}
    \mathcal{N}\left(0, (\Sigma_W, \Sigma_U)\right)
\end{equation*}
We further apply the $\delta$-method, which results in
\begin{gather*}
    \left(\hat W_{n_1, \ldots, n_{t}}^{(t)}\right)^\top \hat U_{n_{t+1}}^{(t)} - 
    \left(W_{n_1, \ldots, n_t}^{(t)}\right)^\top U_{n_{t+1}}^{(t)} \xrightarrow{D} \\
    \mathcal{N}\left(0, \left(W_{n_1, \ldots, n_t}^{(t)}\right)^\top\Sigma_WW_{n_1, \ldots, n_t}^{(t)} + 
    (U_{n_{t+1}}^{(t)})^\top\Sigma_UU_{n_{t+1}}^{(t)}\right)
\end{gather*}
Namely,
\begin{equation*}
    \hat Z_{n_1, \ldots, n_{t+1}} - \mu^{(t+1)} \xrightarrow{D} \mathcal{N}\left(0, Var\left[\hat Z_{n_1, \ldots, n_{t+1}}\right]\right)
\end{equation*}
\end{proof}

\subsection{Generalized Version of Lemma \ref{lemma:var-ub} and Its Proof} \label{proof:var-ub}

\begin{lemma}
    Without loss of generality, we suppose $T_1$ is sampled with sampling rate 
    $\theta_p$ in the pilot query. Then, the standard estimator of the \asum 
    aggregate obtained from a query with a sampling plan $\Theta=[\theta_1, \ldots, 
    \theta_k]$ has a variance of at most
    \begin{align*}
    U_V[\Theta] 
    =  &\sum_{S \subseteq \{1, \ldots, k\} | 1 \in S}
        c_S \cdot U_{y_{S,i}^{(1)}} [\delta] \\
        +&\sum_{S \subseteq \{2, \ldots, k\} | S \neq \emptyset}
        c_S \cdot 
        \sum_{t \in \Lambda(S) }
        \left(U_{y_{t,S,i}^{(2)}}[\delta] \right)^2
    \end{align*}
    with a probability of at least $1-\delta$, where
    \begin{align*}
    c_S &= \sum_{S' \subseteq S} (-1)^{|S|+|S'|} \prod_{i=1}^k \theta_i^{-\mathbbm{1}(i \notin S')}, \\
    y^{(1)}_{S,i} &= \sum_{t_{j,q_j} \in T_j | j \in S,j\neq 1} 
        \left(\sum_{t_{j,q_j} \in T_j | j \in S^C} \myjoin(t_{1,i}, \ldots, t_{k,q_k})\right) ^ 2, \\
    y^{(2)}_{t,S,i} &= \sum_{t' \in \Lambda(S^C)} \myjoin(t_{1,i}, t, t'), \\
    S^C &= \{1, \ldots, k\} / S, \\
    \Lambda(\{1, \ldots, k'\}) &= \{t|t \in T_1\} \times \ldots \times \{t | t \in T_{k'}\} \\
    U_{y_i}[\delta] &= \frac{1}{\theta_p} \left(\sum_{i=1}^{n_p} y_i + \sqrt{n_p} \cdot \hat\sigma(y_i) \cdot t_{\delta, n_p-1}\right)\\
    \delta' &= \delta \bigl/ \left(2^{k-1} + 2\prod_{i=2}^k N_i\right)
    \end{align*}
\end{lemma}

\begin{proof}
We first calculate the variance exactly.
\begin{align}
&Var\left[\left(\prod_{i=1}^k \frac{1}{\theta_i}\right)
    \sum_{t_1 \in \hat T_1} \cdots \sum_{t_k \in \hat T_k} 
    \mathcal{J}(t_1, \ldots, t_k)\right]\\
=& \sum_{S \subseteq \{1, \ldots, k\}} c_S \sum_{t_{j,q_j} \in T_j | j \in S} 
\left(\sum_{t_{j,q_j} \in T_j | j \in S^C} \myjoin(t_{1,i}, \ldots, t_{k,q_k})\right) ^ 2 \label{eq:proof-var-exact}
\end{align}
where $\hat T_i$ is the block sample of $T_i$ with sampling rate $\theta_i$ and
\begin{equation*}
    c_S = \sum_{S' \subseteq S} (-1)^{|S|+|S'|} \prod_{i=1}^k \theta_i^{-\mathbbm{1}(i \notin S')}
\end{equation*}

We observe that Equation \ref{eq:proof-var-exact} requires execute the full query,
which is prohebitively expensive. Therefore, we use the pilot query which samples
the table $T_1$ with sampling rate $\theta_p$ to estimate the upper bound of the
variance. Based on the exact expression of the variance, we need to estimate the
upper bound of the following item for each $S$, the subset of $\{1, \ldots, k\}$.
\begin{equation}
    Y_S = \sum_{t_{j,q_j} \in T_j | j \in S} \left(\sum_{t_{j,q_j} \in T_j | j \in S^C} \myjoin(t_{1,i}, \ldots, t_{k,q_k})\right) ^ 2
    \label{eq:var-decomp}
\end{equation}
Depending on whether the sampling table $T_1$ is in $S$, we consider the following
two cases.

\minihead{Case 1: $1 \in S$}
In this case, we denote
\begin{equation*}
    y^{(1)}_{S,i} = \sum_{t_{j,q_j} \in T_j | j \in S,j\neq 1} 
        \left(\sum_{t_{j,q_j} \in T_j | j \in S^C} \myjoin(t_{1,i}, \ldots, t_{k,q_k})\right) ^ 2
\end{equation*}
We find that $Y_S$ in Equation \ref{eq:var-decomp} is a sum of $y^{(1)}_{S,i}$. 
We can obtain the upper bound of the sum of random variables with the sum of its 
sample and statistics from Student's t distribution. Namely,
\begin{equation*}
    \mathbb{P}\left[Y_S \le \frac{1}{\theta_p}\left(\sum_{i=1}^{n_p} y^{(1)}_{S,i} + \sqrt{n_p} \cdot \hat\sigma(y^{(1)}_{S,i}) \cdot t_{1-\delta, n_p-1}\right)\right] \ge 1 - \delta
\end{equation*}
where $\theta_p$ is the sampling rate, $n_p$ is the size of resulting sample, 
$\hat\sigma(y^{(1)}_{S,i})$ is the standard deviation of the sample, and $t_{1-\delta, n_p-1}$
is the $1-\delta$ percentile of the Student's t distribution with $n_p-1$ degrees 
of freedom.

\minihead{Case 2: $1 \notin S$}
In this case, we denote
\begin{equation*}
    y^{(2)}_{t,S,i} = \sum_{t' \in \Lambda(S^C)} \myjoin(t_{1,i}, t, t')
\end{equation*}
Then we can express $Y_S$ equivalently as
\begin{equation*}
    Y_S = \sum_{t \in \Lambda(S) } \left(\sum_{i=1}^{N_1} y^{(2)}_{t,S,i}\right) ^2
\end{equation*}
which intuitively is summation of squared summation of $y^{(2)}$.
Therefore, we can obtain the upper bound of $Y_S$ by obtaining the upper bound of
each $\sum_{i=1}^{N_1} y^{(2)}_{t,S,i}$, using the same technique in the Case 1.
Namely,
\begin{equation*}
    \mathbb{P}\left[\sum_{i=1}^{N_1} y^{(2)}_{t,S,i} \le \frac{1}{\theta_p}\left(\sum_{i=1}^{n_p} y^{(2)}_{t,S,i} + \sqrt{n_p} \cdot \hat\sigma(y^{(2)}_{t,S,i}) \cdot t_{1-\delta, n_p-1}\right)\right] \ge 1 - \delta
\end{equation*}

Combining two cases and dividing the overall failure probability by the number of
individual upper bounds we used in the derivation, we can obtain the final result.
\end{proof}

\color{black}
\subsection{Comparing Equivalence Rules with Dominance Rules} \label{subsec:quickr-rules}
We first review the definition of sampling dominance \cite{quickr}.
\begin{definition}
Given the same original query, the sampling procedure $\mathcal{S}_1$ with 
output $R_1$ dominates the sampling procedure $\mathcal{S}_2$ with output 
$R_2$, or $\mathcal{S}_1 \xRightarrow{Q} \mathcal{S}_2$, if and only if,
\begin{align*}
    & (v-dominance)\quad \forall i, j:\\
    & \quad \frac{\mathbb{P}\left[i\in R_1, j\in R_1\right]}
        {\mathbb{P}\left[i\in R_1\right]\mathbb{P}\left[j\in R_1\right]} \ge 
        \frac{\mathbb{P}\left[i\in R_2, j\in R_2\right]}
        {\mathbb{P}\left[i\in R_2\right]\mathbb{P}\left[j\in R_2\right]}, and \\
    & (c-dominance)\quad \mathbb{P}\left[t\in R_1\right] \le 
        \mathbb{P}\left[t\in R_2\right]
\end{align*}
\end{definition}

We then prove that our equivalence rules are stronger than the sampling dominance
rules of \textsc{QuickR}.
\begin{theorem}
Given the same origianl query with $k$ input tables $\{T_1, \ldots, T_k\}$, if 
$\mathcal{S}_1 \Leftrightarrow \mathcal{S}_2$, 
then $\mathcal{S}_1 \xRightarrow{Q} \mathcal{S}_2$ and 
$\mathcal{S}_1 \xRightarrow{Q} \mathcal{S}_2$.
\end{theorem}
\begin{proof}
Let $R_1$ be the output of $\mathcal{S}_1$ and $R_2$ be the output of 
$\mathcal{S}_2$. Then, we observe that
\begin{align*}
    \mathbb{P}\left[i \in R_1\right] = \sum_{R \in \{R': i \in R'\}} 
    \mathbb{P}\left[\mathcal{S}_1\left(\{T_1, \ldots, T_k\}\right) = R\right] \\
    \mathbb{P}\left[i \in R_2\right] = \sum_{R \in \{R': i \in R'\}} 
    \mathbb{P}\left[\mathcal{S}_2\left(\{T_1, \ldots, T_k\}\right) = R\right]
\end{align*}
According to the definition of $\mathcal{S}_1 \Leftrightarrow \mathcal{S}_2$,
\begin{equation}
    \forall R,\ \mathbb{P}\left[\mathcal{S}_1\left(\{T_1, \ldots, T_k\}\right)=R\right] 
    = \mathbb{P}\left[\mathcal{S}_2\left(\{T_1, \ldots, T_k\}\right)=R\right]. \notag
\end{equation}
Thus, 
\begin{equation}
    \mathbb{P}\left[i \in R_1\right] = \mathbb{P}\left[i \in R_2\right] \label{eq:prove-c-dominance}
\end{equation}
which proves the mutual c-dominance between $\mathcal{S}_1$ and $\mathcal{S}_1$.

Similarly, we can show that
\begin{align*}
    \mathbb{P}\left[i \in R_1, j \in R_1\right] = \sum_{R \in \{R': i \in R', j\in R'\}} 
    \mathbb{P}\left[\mathcal{S}_1\left(\{T_1, \ldots, T_k\}\right) = R\right] \\
    \mathbb{P}\left[i \in R_2, j \in R_2\right] = \sum_{R \in \{R': i \in R', j\in R'\}} 
    \mathbb{P}\left[\mathcal{S}_2\left(\{T_1, \ldots, T_k\}\right) = R\right]
\end{align*}
Since $\mathcal{S}_1 \Leftrightarrow \mathcal{S}_2$, we have
\begin{equation}
    \mathbb{P}\left[i \in R_1, j \in R_2\right] = \mathbb{P}\left[i \in R_2, j \in R_2\right] \label{eq:prove-v-dominance}
\end{equation}
Equation \ref{eq:prove-c-dominance} and \ref{eq:prove-v-dominance} together 
prove the mutual v-dominance between $\mathcal{S}_1$ and $\mathcal{S}_1$. Hence,
$\mathcal{S}_1 \xRightarrow{Q} \mathcal{S}_2$ and 
$\mathcal{S}_1 \xRightarrow{Q} \mathcal{S}_2$.
\end{proof}

As shown, the sampling dominance rules of \textsc{QuickR} only considers the 
inclusion probability of one or two sampling units, which does not prove the
equivalence. The equivalence property enables \bstats to not only maintain error
guarantees but also avoid unnecessary increasing of sampling rates.
\color{black}

\color{black}
\section{Supported Database Management Systems} \label{sec:app-dbms}
Existing databases support sampling at the page level, shard level, or row level. 
\name can accelerate queries for those supporting page-level sampling. These 
databases include PostgreSQL, SQL Server, and Oracle. However, page-level 
sampling can be implemented in other databases, such as Snowflake and BigQuery easily. We categorize databases as follows:
\begin{enumerate}[leftmargin=*]
    \item Databases that support page-level sampling: PostgreSQL, SQL Server, Oracle, Apache Hive, DB2, DuckDB
    \item Databases that only support shard-level sampling: Snowflake, BigQuery
    \item Databases that only support row-level sampling: SparkSQL, MongoDB
    \item Databases that do not support sampling: MySQL
\end{enumerate}

\color{black}

\end{document}